\documentclass[12pt]{article}
\usepackage[latin1]{inputenc}
\usepackage{amssymb}
\usepackage{amsmath}
\usepackage{amsmath}
\usepackage{amsthm}
\usepackage{amsfonts}
\usepackage{graphicx}
\usepackage{enumerate}
\usepackage{dsfont}
\usepackage{epsfig}
\usepackage{color}

\numberwithin{equation}{section}
\newtheorem{theorem}{Theorem}

\newtheorem{algorithm}{Algorithm}

\newtheorem{corollary}[theorem]{Corollary}

\newtheorem{proposition}{Proposition}

\newtheorem{lemma}{Lemma}
\newtheorem{ass}{Assumption}
\theoremstyle{definition}

\newcommand{\bU}{{\boldsymbol{U}}}

\renewcommand{\Pr}{\mathbb{P}}

\begin{document}

\title{Nonparametric Identification in Panels using Quantiles \thanks{%
We thank the editor Qi Li, two anonymous referees, participants at Demand
Estimation and Modelling Conference, and Arthur Lewbel for comments. We
gratefully acknowledge research support from the NSF.}}
\author{Victor Chernozhukov, Ivan Fernandez-Val, Stefan Hoderlein, \\
Hajo Holzmann, Whitney Newey \ }
\date{\today }
\maketitle

\begin{abstract}
This paper considers identification and estimation of ceteris paribus
effects of continuous regressors in nonseparable panel models with time
homogeneity. The effects of interest are derivatives of the average and
quantile structural functions of the model. We find that these derivatives
are identified with two time periods for ``stayers", i.e. for individuals
with the same regressor values in two time periods. We show that the
identification results carry over to models that allow location and scale
time effects. We propose nonparametric series methods and a weighted
bootstrap scheme to estimate and make inference on the identified effects.
The bootstrap proposed allows uniform inference for function-valued
parameters such as quantile effects uniformly over a region of
quantile indices and/or regressor values. An empirical application to Engel curve estimation with panel data
illustrates the results.
\end{abstract}


\section{Identification for Panel Regression}

A frequent object of interest is the ceteris paribus effect of $x$ on $y,$
when observed $x$ is an individual choice variable partly determined by
preferences or technology. Panel data holds out the hope of controlling for
individual preferences or technology by using multiple observations for a
single economic agent. This hope is particularly difficult to realize with
discrete or other nonseparable models and/or multidimensional individual
effects. These models are, by nature, not additively separable in unobserved
individual effects, making them challenging to identify and estimate.

A fundamental idea for using panel data to identify the ceteris paribus
effect of $x$ on $y$ is to use changes in $x$ over time. In order for
changes over time in $x$ to correspond to ceteris paribus effects, the
distribution of variables other than $x$ must not vary over time. This
restriction is like \textquotedblleft time being randomly
assigned\textquotedblright\ or "time is an instrument."\ In this paper we
consider identification via such time homogeneity conditions. They are also
the basis of many previous panel results, including Chamberlain (1982),
Manski (1987), and Honore (1992). Recently time homogeneity has been used as
the basis for identification and estimation of nonseparable models by
Chernozhukov, Fernandez-Val, Hahn, Newey (2013), Evdokimov (2010), Graham
and Powell (2012), and Hoderlein and White (2012). Because economic data
often exhibits drift over time, we also allow for some time effects, while
maintaining underlying time homogeneity conditions.

In this paper we give identification and estimation results for quantile
effects with time homogeneity and continuous regressors. The effects of
interest are derivatives of quantile structural functions of the model. We
find that these derivatives are identified with two time periods for
\textquotedblleft stayers", i.e. conditional on $x$ being equal in two time
periods. Time homogeneity is too strong for many econometric applications
where time trends are evident in the data. We weaken homogeneity by allowing
for location and scale time effects. Allowing for such time effects makes
identification and estimation more complicated but more widely applicable.
We also give analogous results for conditional mean effects under weaker
identification conditions than previously.

Quantile identification under time homogeneity is based on differences of
quantiles. It is also interesting to consider whether quantiles of
differences can help identify effects of interest. We do not find that time
homogeneity alone can lead to identification from quantiles of differences.
We do give quantile difference identification results that restrict the
distribution of individual effects conditional on $x$, similarly to
Chamberlain (1980), Altonji and Matzkin (2005), and Bester and Hansen
(2009). In our opinion these added restrictions make quantiles of
differences less appealing. We therefore focus for the rest of the paper,
including the application, on differences of quantiles.

To illustrate we provide an application to Engel curve estimation. The Engel
curve describes how demand changes with expenditure. We use data from the
2007 and 2009 waves of the Panel Study of Income Dynamics (PSID).
Endogeneity in the estimation of Engel curves arises because the decision to
consume a commodity may occur simultaneously with the allocation of income
between consumption and savings. In contrast with the previous cross
sectional literature, we do not rely on a two-stage budgeting argument that
justifies the use of labor income as an instrument for expenditure. Instead,
we assume that the Engel curve relationships are time homogeneous up to
location and scale time effects, which leads to identification of structural
effects from panel data.

An alternative approach to identification in panel data is to impose
restrictions on the conditional distribution of the individual effect given $%
x$. This approach leads to nonparametric generalizations of Chamberlain's
(1980) correlated random effects model. As shown by Chamberlain (1984),
Altonji and Matzkin (2005), Bester and Hansen (2009), and others, this kind
of condition leads to identification of various effects. In particular,
Altonji and Matzkin (2005) show identification of an average derivative
conditional on the regressor equal to a specific value, an effect they call
the local average response (LAR). In this paper we take a different
approach, preferring to impose time homogeneity rather than restrict the
relationship between observed regressors and unobserved individual effects.
We refer to Hsiao (2003) for a broader perspective of panel data models.

Section 2 describes the model and gives an average derivative result.
Section 3 gives the quantile identification result that follows from time
homogeneity. Section 4 considers how quantiles of differences can be used to
identify the effect of $x$ on $y$. Section 5 explains how we allow for time
effects. Estimation and inference are briefly discussed in Section 6, and
the empirical example is given in Section 7. The Appendix contains the
proofs of the main results.

\section{The Model and Conditional Mean Effects}

The data consist of $n$ observations on ${\boldsymbol{Y}}_{i}=(Y_{i1},%
\ldots,Y_{iT})^{\prime }$ and ${\boldsymbol{X}}_{i}=[X_{i1}^{\prime
},\ldots,X_{iT}^{\prime}] ^{\prime }$, for a dependent variable $Y_{it}$ and
a vector of regressors $X_{it}$. Throughout we assume that the observations $%
({\boldsymbol{Y}}_{i},{\boldsymbol{X}}_{i})$\textit{, }$(i=1,\ldots,n)$%
\textit{, }are independent and identically distributed. The nonparametric
models we consider satisfy

\bigskip

\begin{ass}
\label{assum:model} There is a function $\phi$ and vectors of random
variables $A_{i}$ and $V_{it}$ such that 
\begin{equation*}
Y_{it}=\phi (X_{it},A_{i},V_{it}),\qquad i=1,\ldots,n, \quad t=1,2, \ldots,
T.
\end{equation*}
\end{ass}

We focus in this paper on the two time period case, $T=2$, though it is
straightforward to extend the results to many time periods. The vector $%
A_{i} $ consists of time invariant individual effects that often represent
individual heterogeneity. The vector $V_{it}$ represents period specific
disturbances. Altonji and Matzkin (2005) considered models satisfying
Assumption 1. As discussed in Chernozhukov et. al. (2013), the invariance of 
$\phi $ over time in this Assumption does not actually impose any time
homogeneity. If there are no restrictions on $V_{it}$ then $t$ could be one
of the components of $V_{it},$ allowing the function to vary over time in a
completely general way. The next condition together with Assumption 1
imposes time homogeneity on the model.

\bigskip

\begin{ass}
\label{ass:distrinv} $V_{it}|{\boldsymbol{X}}_{i},A_{i}\overset{d}{=}V_{i1}|{%
\boldsymbol{X}}_{i},A_{i}$ \quad for all $t.$
\end{ass}

\bigskip

This is a static, or "strictly exogenous" time homogeneity condition, where
all leads and lags of the regressor are included in the conditioning
variable ${\boldsymbol{X}}_{i}.$ It requires that the conditional
distribution of $V_{it}$ given ${\boldsymbol{X}}_{i}$ and $A_{i}$ does not
depend on $t,$ but does allow for dependence of $V_{it}$ over time. This
assumption rules out dynamic models where lagged values of $Y_{it}$ are
included in $X_{it}.$

Setting $U_{it}=(A_{i}^{\prime },V_{it}^{\prime })^{\prime }$, an equivalent
condition is 
\begin{equation*}
U_{it}|{\boldsymbol{X}}_{i}\overset{d}{=}U_{i1}|{\boldsymbol{X}}_{i}.
\end{equation*}%
Thus, the time invariant $A_{i}$ has no distinct role in this model. As
further discussed in Chernozhukov et. al. (2013), this seems a basic
condition that helps panel data provide information about the effect of $x$
on $y.$ It is like the time period being "randomly assigned" or "time is an
instrument," with the distribution of factors other than $x$ not varying
over time, so that changes in $x$ over time can help identify the effect of $%
x$ on $y$.

Although they seem useful for nonlinear models, the time homogeneity
conditions are strong. In particular they do not allow for
heteroskedasticity over time, which is often thought to be important in
applications. We partially address this problem below by allowing for
location and scale time effects.

For notational convenience we shall drop the $i$ subscript and let $T=2$ in
the following. Our focus in this paper is on the case where the regressors ${%
\boldsymbol{X}}$ are continuously distributed. We will be interested in
several effects of ${\boldsymbol{X}}$ on ${\boldsymbol{Y}}$. For $%
u=(a^{\prime },v^{\prime })^{\prime }$ we let $\phi (x,u)=\phi (x,a,v)$. We
will let $x$ or $x_{t}$ denote a possible value of the regressor vector $%
X_{t}$ and ${\boldsymbol{x}}=(x_{1}^{\prime },x_{2}^{\prime })^{\prime }$ a
possible value of ${\boldsymbol{X}}=(X_{1}^{\prime },X_{2}^{\prime
})^{\prime }$. Let $\partial _{x}\phi (x,u)$ denote the vector of partial
derivatives of $\phi $ w.r.t.~the coordinates of $x$. One effect we consider
is a conditional expectation of the derivative $\partial _{x}\phi
(X_{t},U_{t})$ given by 
\begin{equation*}
E\big[\partial _{x}\phi (x_{,}U_{t})|X_{1}=X_{2}=x\big].
\end{equation*}%
This is the object considered in Hoderlein and White (2012) and is similar
to the local average response considered in Altonji and Matzkin (2005). It
gives the local marginal effect for individuals with regressor value $x$ in
both periods. This effect is related to the conditional average structural
function (CASF): 
\begin{equation*}
m(x\mid {\boldsymbol{x}})=E[\phi (x,U_{t})\mid {\boldsymbol{X}}={\boldsymbol{%
x}}],
\end{equation*}%
through 
\begin{equation*}
\partial _{x}m(x\mid {\boldsymbol{x}})\Big|_{{\boldsymbol{x}}=(x,x)}=E\big[%
\partial _{x}\phi (x_{,}U_{t})|X_{1}=X_{2}=x\big],
\end{equation*}%
under the conditions that permit interchanging the derivative and
expectation.

The other effects we consider are similar to this effect except that we also
condition on certain values of $Y_{t}$. One of these is given by 
\begin{equation*}
E\big[\partial _{x}\phi (x,U_{t})|Y_{t}=q(\tau ,x),X_{1}=X_{2}=x\big],
\end{equation*}%
where $q(\tau ,x)$ is the $\tau ^{th}$ conditional quantile of $\phi
(x,U_{t})$ given $X_{1}=X_{2}=x.$ This is a quantile derivative effect,
similar to the local average structural derivative in Hoderlein and Mammen
(2007). It gives the local marginal effect for individuals with regressor
value $x$ in both periods and at the quantile $q(\tau ,x)$. This effect is
also related to the conditional quantile structural function (CQSF), $%
q_{\tau }(x|{\boldsymbol{x}})$, that gives the $\tau $-quantile of $\phi
(x,U_{t})$ conditional on ${\boldsymbol{X}}={\boldsymbol{x}},$ through 
\begin{equation*}
\partial _{x}q_{\tau }(x\mid {\boldsymbol{x}})\Big|_{{\boldsymbol{x}}%
=(x,x)}=E\big[\partial _{x}\phi (x,U_{t})|Y_{t}=q(\tau ,x),X_{1}=X_{2}=x\big]%
.
\end{equation*}

We also consider linking quantiles of arbitrary linear combinations of the
dependent variables $Y_{1}$ and $Y_{2}$ to conditional expectations of the
form 
\begin{equation*}
E\big( \partial_x \phi (x,U_{t})\,|\,\text{linear comb of } {\boldsymbol{Y}}%
, X_{1}=X_{2}=x\big) .
\end{equation*}
These are dependent variable conditioned average effects. One intended
direction is to compare the derivative of the quantiles of the differences $%
Y_{2}-Y_{1}$ to the differences of the derivative of the quantiles of $Y_{2}$
and $Y_{1}$ in terms of objects they identify.  In what follows we carry out the comparison.

To set the stage for the quantile results we first discuss mean
identification. We first give an explanation of identification of the mean
effect and then give a precise result with regularity conditions.

Consider the identified conditional mean 
\begin{equation*}
M_{t}({\boldsymbol{x}})=E(Y_{t}|{\boldsymbol{X}}={\boldsymbol{x}}),\qquad
t=1,2.
\end{equation*}%
Together these conditional expectations are a nonparametric version of
Chamberlain's (1982) multivariate regression model for panel data.
Derivatives of them can be combined to identify the conditional mean effect.
Let $f(u|{\boldsymbol{x}})$ denote the conditional density of $U_{t}$ given $%
{\boldsymbol{X}}={\boldsymbol{x}},$ that does not depend on $t$ by
Assumption \ref{ass:distrinv}. Assume that $\phi (x,u)$ and $f(u|{%
\boldsymbol{x}})$ are differentiable in $x$ and ${\boldsymbol{x}}$
respectively and that differentiation under the integral is permitted. For ${%
\boldsymbol{x}}=(x_{1}^{\prime },x_{2}^{\prime })^{\prime }$ we let $%
\partial _{x_{s}}M_{t}({\boldsymbol{x}})$ and $\partial _{x_{s}}f(u|{%
\boldsymbol{x}})$, $s,t=1,2$, denote the vector of partial derivatives
w.r.t.~the coordinates of $x_{s}$. Then for $s,t=1,2$, 
\begin{align*}
\partial _{x_{s}}M_{t}({\boldsymbol{x}})& =\partial _{x_{s}}E(Y_{t}|{%
\boldsymbol{X}}={\boldsymbol{x}})=\partial _{x_{s}}\int \phi (x_{t},u)f(u|{%
\boldsymbol{x}})du \\
& =1(s=t)\int \partial _{x}\phi (x_{t},u)f(u|{\boldsymbol{x}})du+\int \phi
(x_{t},u)\partial _{x_{s}}f(u|{\boldsymbol{x}})du,
\end{align*}%
where the first term is the conditional mean effect of interest and the
second term is the analog to Chamberlain's (1982) heterogeneity bias.
Subtracting and using Assumption \ref{ass:distrinv} gives 
\begin{equation}
\partial _{x_{2}}M_{2}({\boldsymbol{x}})-\partial _{x_{2}}M_{1}({\boldsymbol{%
x}})=E\big(\partial _{x}\phi (x_{2},U_{t})|{\boldsymbol{X}}={\boldsymbol{x}}%
\big)+\int \big(\phi (x_{2},u)-\phi (x_{1},u)\big)\partial _{x_{2}}f(u|{%
\boldsymbol{x}})du.  \label{cond mean eq}
\end{equation}%
Evaluating at ${\boldsymbol{x}}=(x^{\prime },x^{\prime })^{\prime }$ we find
that%
\begin{equation}
E\big(\partial _{x}\phi (x,U_{t})|X_{1}=X_{2}=x\big)=\partial
_{x_{2}}M_{2}(x,x)-\partial _{x_{2}}M_{1}(x,x)=\partial _{x_{2}}\Delta M(x,x)
\label{eq:resultmean1}
\end{equation}%
where 
\begin{equation*}
\Delta M({\boldsymbol{x}})=E(Y_{2}-Y_{1}|{\boldsymbol{X}}={\boldsymbol{x}}).
\end{equation*}%
It also follows similarly that 
\begin{equation}
\begin{split}
E\big(\partial _{x}\phi (x,U_{t})|X_{1}=X_{2}=x\big)& =-\partial
_{x_{1}}\Delta M(x,x) \\
& =\partial _{x_{1}}E(Y_{1}-Y_{2}|X_{1}=x_{1},X_{2}=x_{2})\big|%
_{(x_{1}^{\prime },x_{2}^{\prime })=(x^{\prime },x^{\prime })}.
\end{split}
\label{eq:resultmean2}
\end{equation}%
Thus, the conditional mean effect is identified from the derivative of the
conditional expectation of the difference with respect to the leading time
period for individuals where $X_{t}$ is the same in both periods. We note
here that this means the conditional mean effect is overidentified.
Introducting time effects, as we do below, will lead to exact
identification. Thus, testing for the presence of time effects is one way of
testing this overidentifying restriction.

The importance of conditioning on the event $x=X_{1}=X_{2}$ can be seen from
equation (\ref{cond mean eq}), where setting $X_{1}=X_{2}$ eliminates
heterogeneity bias. Thus, one can think of the conditioning on $X_{1}=X_{2}$
as a device to eliminate the heterogenity bias in nonseparable models under
time stationarity. In contrast, if $\phi (x,u)$ were additively separable
with $\phi (x,u)=\mu (x)+u,$ the heterogeneity bias would be zero for all $%
X_{1}$ not necessarily equal to $X_{2}$ because $\int \partial _{x_{2}}f(u|{%
\boldsymbol{x}})du=0$. Hence the derivative effect of interest would be $%
\partial_{x_2} \Delta M(\boldsymbol{x})$ for each value of $x_{1}$ and one
could estimate that derivative more precisely by averaging over its first
argument. Also, one could test for whether the model is additively separable
by testing whether $\Delta M(\boldsymbol{x})$ varies with its first
argument, though it is beyond the scope of this paper to analyze such tests.

Conditioning on $x=X_{1}=X_{2}$ does restrict the set over which the
structural derivative is averaged but this can correspond to an interesting
set of individuals. For example, in the Engel curve application we give $x$
is total expenditure so the restriction $X_{1}=X_{2}$ corresponds to
individuals whose total expenditure was the same in the two time periods.
This seems mostly likely to occur for middle aged individuals, which is an
interesting though special group to focus on.

Altonji and Matzkin (2005) are able to identify derivative effects without
conditioning on $X_{1}=X_{2}$ but they also restrict the distribution of $%
U_{t}$ conditional on $\boldsymbol{X}.$ We do not impose such type of
assumptions but instead require time stationarity of the distribution of $%
U_{t}$ conditional on $\boldsymbol{X}$. The different assumptions make it
hard to compare results. We prefer to focus on time stationarity in this
paper, where we do not yet know whether it is possible to identify
interesting effects for continuous regressors without imposing $X_{1}=X_{2}.$

Graham and Powell (2012) consider a linear model with individual specific
coefficients where $\phi (x,u)=\beta _{1}(u)+\beta _{2}(u)x$ in the scalar $x
$ case. In this case%
\begin{equation*}
E\big[\partial _{x}\phi (x_{,}U_{t})|X_{1}=X_{2}=x\big]=E\big[\beta
_{2}(U_{t})|X_{1}=X_{2}=x\big]
\end{equation*}%
Here we find that average slope for the stayer subpopulation with $%
X_{1}=X_{2}$ is identified. Graham and Powell (2012) use linearity of $\phi
(x,u)$ in $x$ to identify the average slope $E\big[\beta _{2}(U_{t})\big]$
over the whole population using the movers with $X_{1}\neq X_{2}$. We
identify an average slope over a smaller population for a fully nonlinear,
nonparametric specification $\phi (x,u)$.

The following result makes the previous derivation precise, including
conditions for differentiating under integrals.

\begin{theorem}
\label{theorem:mean} Suppose that Assumptions \ref{assum:model} and \ref%
{ass:distrinv} are satisfied, $E|Y_{t}|<\infty$, $t=1,2,$ and that $\phi
(x,u)$ (where $u^{\prime }= (a^{\prime }, v^{\prime })$) resp.~the
conditional density $f(u|{\boldsymbol{x}})$ of $U_{t}=(A^{\prime },
V_{t}^{\prime })^{\prime }$ given ${\boldsymbol{X}}={\boldsymbol{x}}$ are
continuously differentiable in $x$ resp.~${\boldsymbol{x}}$ for fixed $u$.
Given $x$, suppose that for some $\varepsilon >0$, 
\begin{align*}
& \int \sup_{\left\Vert \delta \right\Vert \leq \varepsilon, \delta =
(\delta_0^{\prime }, \delta_1^{\prime }, \delta_2^{\prime })^{\prime }}%
\big\Vert \partial_x \phi(x+\delta_0 ,u)\, f(u|x+\delta_{1},x+\delta _{2}) %
\big\Vert du<\infty , \\
& \int \sup_{\left\Vert \delta \right\Vert \leq \varepsilon, \delta =
(\delta_0^{\prime }, \delta_1^{\prime }, \delta_2^{\prime })^{\prime }}%
\big\Vert \phi(x+\delta_0 ,u)\, \partial_{x_s} f(u|x+\delta_{1},x+\delta
_{2}) \big\Vert du<\infty, \quad s=1,2,
\end{align*}
then (\ref{eq:resultmean1}) and (\ref{eq:resultmean2}) hold true.
\end{theorem}

%

This result has slightly weaker conditions than that of Hoderlein and White
(2012). Here we drop their assumption that $V_{t}$ is independent of $X_{1}$
conditional on $A$. The result given here allows for $X_{1}$ to be
correlated with $(V_{1},V_{2}),$ as long as the marginal distribution of $%
V_{t}$ conditional on $(X_{1},X_{2},A )$ does not vary with $t.$ We maintain
these weaker conditions as we consider identification of quantile effects in
the next Section.

\section{Conditional Quantile Effects}

\label{sec:quant}

Turning now to the identification of the quantile effects given above, let $%
Q_{t}(\tau \mid {\boldsymbol{x}})$ denote the $\tau ^{th}$ conditional
quantile of $Y_{t}$ conditional on ${\boldsymbol{X}}={\boldsymbol{x}}%
=(x_{1}^{\prime },x_{2}^{\prime })^{\prime }$. It will be the solution to 
\begin{equation*}
\int 1(\phi (x_{t},u)\leq Q_{t}(\tau \mid {\boldsymbol{x}}))f(u|{\boldsymbol{%
x}})du=\tau .
\end{equation*}%
The pair $[Q_{1}(\tau | {\boldsymbol{x}}),Q_{2}(\tau | {\boldsymbol{x}})]$
is a quantile analog of Chamberlain's (1982) multivariate regression for
panel data. We can identify a quantile analog of the Hoderlein and White
(2012) average derivative effect. We first describe how these multivariate
panel quantiles can be used to identify an average derivative effect, then
give a precise interpretation of the effect. This description helps explain
the source of identification as well as the precise nature of the identified
effect.

To describe how identification works, differentiate both sides of the
previous identity with respect to $x_{s}$, treat the derivative of an
indicator function as a dirac delta, and assume the order of differentiation
and integration can be interchanged. This calculation gives 
\begin{align*}
0& =\int_{\phi (x_{t},u)=Q_{t}(\tau \mid {\boldsymbol{x}})}\big(\partial
_{x_{s}}Q_{t}(\tau |{\boldsymbol{x}})-1(s=t)\partial _{x}\phi (x_{t},u)\big)%
f(u|{\boldsymbol{x}})du \\
& +\int 1\big(\phi (x_{t},u)\leq Q_{t}(\tau \mid {\boldsymbol{x}})\big)%
\partial _{x_{s}}f(u|{\boldsymbol{x}})du.
\end{align*}%
Let $g_{t}(\tau \mid {\boldsymbol{x}})=\int_{\phi (x_{t},u)=Q_{t}(\tau \mid {%
\boldsymbol{x}})}f(u|{\boldsymbol{x}})du$ and note that 
\begin{equation*}
g_{t}(\tau \mid {\boldsymbol{x}})^{-1}\int_{\phi (x_{t},u)=Q_{t}(\tau \mid {%
\boldsymbol{x}})}\partial _{x}\phi (x_{t},u)f(u|{\boldsymbol{x}})du=E\big(%
\partial _{x}\phi (x_{t},U_{t})|\phi (x_{t},U_{t})=Q_{t}(\tau \mid {%
\boldsymbol{x}}),{\boldsymbol{X}}={\boldsymbol{x}}\big)
\end{equation*}%
Solving for $\partial _{x_{s}}Q_{t}(\tau |{\boldsymbol{x}})$ we find that, 
\begin{align*}
\partial _{x_{s}}Q_{t}(\tau |{\boldsymbol{x}})& =1(s=t)E\big(\partial
_{x}\phi (x_{t},U_{t})|\phi (x_{t},U_{t})=Q_{t}(\tau \mid {\boldsymbol{x}}),{%
\boldsymbol{X}}={\boldsymbol{x}}\big) \\
& -g_{t}(\tau \mid {\boldsymbol{x}})^{-1}\int 1(\phi (x_{t},u)\leq
Q_{t}(\tau \mid {\boldsymbol{x}}))\partial _{x_{s}}f(u|{\boldsymbol{x}})du.
\end{align*}%
Note that at $X_{1}=X_{2}=x$, $Q_{1}(\tau \mid x,x)=Q_{2}(\tau \mid
x,x)=q(\tau ,x)$ and $g_{1}(\tau \mid x,x)=g_{2}(\tau \mid x,x)$ by time
homogeneity. Then differencing the conditional quantile derivatives gives 
\begin{equation}  \label{eq:quantileeffect}
\begin{split}
\partial _{x_{2}}Q_{2}(\tau |x,x)-\partial _{x_{2}}Q_{1}(\tau |x,x)&
=\partial _{x_{1}}Q_{1}(\tau |x,x)-\partial _{x_{1}}Q_{2}(\tau |x,x) \\
& =E\big(\partial _{x}\phi (x,U_{t})|\phi (x,U_{t})=q(\tau ,x),X_{1}=X_{2}=x%
\big),
\end{split}%
\end{equation}%
where the last term does not depend on $t$ due to time homogeneity. The
equation (\ref{eq:quantileeffect}) is a panel version of the Hoderlein and
Mammen (2007) identification result. It is interesting to note that, unlike
in the mean case, differences of derivatives of quantiles generally differ
from derivatives of quantiles of differences. Below we will consider
identification from derivatives of quantiles of differences.

To make the above derivation precise we need to formulate conditions that
allow differentiation under the integral. The following regularity condition
is one approach to this, in particular for the dirac delta argument given
above.

\begin{ass}
\label{ass:regquantile} We can write $u=(h^{\prime },e )^{\prime }$ for
scalar $e$ , such that $\phi (x,u)=\phi (x,h ,e )$ is continuously
differentiable in $x$ and $e$ and there is $C>0$ with $\partial_e \phi (x,h
,e ) \geq 1/C$ and $\left\Vert \partial_x \phi (x,u)\right\Vert \leq C$
everywhere. For the corresponding representation of the random vector $%
U_t=(H_t^{\prime },E_t)$, $E_t$ is continuously distributed given $(H_t ,{%
\boldsymbol{X}}),$ with conditional pdf $f_E (e |h ,{\boldsymbol{x}})$ that
is bounded and continuous in $(e ,{\boldsymbol{x}}),$ and $f(h |{\boldsymbol{%
x}})$, the conditional pdf of $H$ given ${\boldsymbol{X}} = \boldsymbol{x}$,
is continuous in ${\boldsymbol{x}}$. Moreover, given ${\boldsymbol{x}}$
there is a $\delta >0$ such that 
\begin{equation}  \label{eq:domcond}
\int \sup_{\left\Vert \Delta _{{\boldsymbol{x}}}\right\Vert \leq \delta }f(h
|{\boldsymbol{x}}+\Delta _{{\boldsymbol{x}}})dh <\infty .
\end{equation}
\end{ass}

The boundedness conditions on the derivatives of $\phi (x,u)$ could further
be weakened at the expense of much more complicated notation and conditions.

For fixed $x$ let $f_{Y_{x}|{\boldsymbol{X}}}(y|{\boldsymbol{x}})$ denote
the conditional pdf of $Y_{x}=\phi (x,U_{t})$ given ${\boldsymbol{X}}={%
\boldsymbol{x}} = (x_1^{\prime }, x_2^{\prime })^{\prime }.$ The following
lemma shows differentiability of $\Pr (\phi (x,U_{t})\leq y|{\boldsymbol{X}}=%
{\boldsymbol{x}})$ with respect to $x$ and $y $ for given ${\boldsymbol{x}}$%
, and computes the derivatives.

\begin{lemma}
\label{lemm:quantilemain} If Assumption \ref{ass:regquantile} is satisfied
then for fixed ${\boldsymbol{x}}$, $\Pr (\phi (x,U_t)\leq y|{\boldsymbol{X}}
= {\boldsymbol{x}})$ is differentiable in $y$ and $x$ with derivatives
continuous in $y$, $x$ and ${\boldsymbol{x}}$ given by 
\begin{align*}
\partial_y \Pr (\phi (x,U_t)\leq y|{\boldsymbol{X}}={\boldsymbol{x}})
&=f_{Y_{x}|{\boldsymbol{X}}}(y|{\boldsymbol{x}}), \\
\partial_x \Pr (\phi (x,U_t)\leq y|{\boldsymbol{X}}={\boldsymbol{x}}) &= -
f_{Y_{x}|{\boldsymbol{X}}}(y|{\boldsymbol{x}})\, E\big(\partial_x \phi
(x,U_t) |Y_{x}=y,{\boldsymbol{X}}={\boldsymbol{x}}\big),
\end{align*}
where $Y_{x}=\phi (x,U_{t})$.
\end{lemma}

With this result in hand we can now make precise the quantile effect
sketched above.

\begin{theorem}
\label{theorem:quant} If Assumptions \ref{assum:model} - \ref%
{ass:regquantile} are satisfied, $f(u|{\boldsymbol{x}})$ is continuously
differentiable in ${\boldsymbol{x}}$, 
\begin{equation}  \label{eq:dommain}
\int \sup_{\left\Vert \Delta _{{\boldsymbol{x}}}\right\Vert \leq
\delta}\left\Vert \partial_{\boldsymbol{x}} f(u|{\boldsymbol{x}}+\Delta _{{%
\boldsymbol{x}}})\right\Vert du<\infty ,
\end{equation}
and the conditional density of $Y_t$ given ${\boldsymbol{X}}$ is positive on
the interior of its support then for all $0<\tau <1,$ $Q_{t}(\tau | {%
\boldsymbol{x}})$ exists and is continuously differentiable such that (\ref%
{eq:quantileeffect}) holds true.
\end{theorem}

To illustrate the previous result, consider the familiar linear model with
additive heterogeneity $Y_t = X_{t}^{\prime }\theta + U_t,$ where $U_t = A +
V_t.$ Let $\overline{Q}_{\tau}(\cdot \mid {\boldsymbol{X}})$ denote the
linear $\tau$-quantile regression on $\text{vec}({\boldsymbol{X}}),$ a
quantile version of the panel multivariate regression of Chamberlain (1982).
Under time homogeneity 
\begin{equation*}
\overline{Q}_{\tau}(Y_t \mid {\boldsymbol{x}}) = x_t^{\prime }\theta + 
\overline{Q}_{\tau}(U_t \mid {\boldsymbol{x}}) = x_t^{\prime }\theta +
x_1^{\prime }\gamma_{\tau1} + x_2^{\prime }\gamma_{\tau2},
\end{equation*}
where $\gamma_{\tau1}$ and $\gamma_{\tau2}$ do not depend on $t$. Taking
derivatives and differencing over time, for $s \neq t,$ 
\begin{equation*}
\partial_{x_t} \overline{Q}_{\tau}(Y_t \mid {\boldsymbol{x}}) -
\partial_{x_t} \overline{Q}_{\tau}(Y_s \mid {\boldsymbol{x}}) = \theta +
\gamma_{\tau t} - \gamma_{\tau t} = \theta.
\end{equation*}
Here the result holds for sequences ${\boldsymbol{x}}$ with $x_1 \neq x_2$
because the heterogeneity is additive.

\section{Quantiles of Transformations of the Dependent Variables}

In this section we answer the question whether we can relate quantiles of
the first difference of the dependent variable to causal effects. In fact,
the same arguments and assumptions that are used for first differences can
also be employed for arbitrary functions of the dependent variables which
map the $T$-vector of dependent variables ${\boldsymbol{Y}}$ (in our case
for simplicity $T=2$) into a scalar \textquotedblleft index". However, as it
turns out, if we restrict ourselves to using only two time periods of the
covariates $X_{t} $, we have to strengthen the assumptions significantly to
make statements about causal effects. This is related to the fact that we do
not have an auxiliary equation at our disposal that allows us to correct for
the heterogeneity bias that arose from the correlation of $X_{t}$ and $%
U_{s}. $

To be more specific about the assumptions: While still considering the model
specified in Assumption \ref{assum:model}, instead of time homogeneity
assumption \ref{ass:distrinv}, in this section we shall use independence
assumptions.

\begin{ass}
\label{assum:condindependent}

\begin{enumerate}
\item $(V_{1},V_{2})\text{ are independent of }(X_{1},X_{2})|A,$

\item $A$ is independent of $X_{2}|X_{1},$
\end{enumerate}
\end{ass}

The first part of this assumption states that the transitory error component
is independent of covariates, given the persistent fixed effect, which is a
notion of strict exogeneity. The second part of this assumption is more
restrictive as it rules out the case where $A$ is arbitrarily correlated
with the $X_{t}$ process. This is a special case of the sufficient statistic
type assumptions in Altonji and Matzkin (2005). Assumption \ref{ass:distrinv}
does not restrict the relationship between $X_t$ and $A$ and allows for $X_t$
and $V_t$ to be correlated, but it is not formally nested within Assumption %
\ref{assum:condindependent}. To see this, consider the example of the panel
multivariate quantile regression in the additive linear model of Section \ref%
{sec:quant}. Without time homogeneity, 
\begin{equation*}
\overline{Q}_{\tau}(Y_t \mid {\boldsymbol{x}}) = x_t^{\prime }\theta +
x_1^{\prime }\gamma_{\tau1,t} + x_2^{\prime }\gamma_{\tau2,t}, \ \ t = 1,2.
\end{equation*}
Assumption \ref{ass:distrinv} imposes time homogeneity on the coefficients,
i.e., $\gamma_{\tau1,t} = \gamma_{\tau1},$ and $\gamma_{\tau2,t} =
\gamma_{\tau2},$ whereas Assumption \ref{assum:condindependent} imposes the
exclusion restrictions $\gamma_{\tau2,1} = 0$ and $\gamma_{\tau2,2} = 0$ but
lets $\gamma_{\tau1,t}$ vary with $t$. In our view these exclusion
restrictions are stronger than time homogeneity in most economic
applications.

\medskip

To adopt a similar framework as above, we rewrite 
\begin{equation*}
{\boldsymbol{U}}=\left( V_{1},V_{2},A\right) ^{T},
\end{equation*}%
and note that the independence and strict exogeneity assumptions imply that:

\begin{lemma}
\label{lem:independ} Under Assumption \ref{assum:condindependent}, ${%
\boldsymbol{U}}$ and $X_2$ are independent given $X_1$.
\end{lemma}

\begin{proof}
For measurable sets $K_i$, $i = 1,2,3,$
\begin{align*}
& P\big(V_1 \in K_1, V_2 \in K_2, A \in K_3 | X_1 = x_1, X_2 = x_2 \big) \\
= & \int_{K_3} P\big(V_1 \in K_1, V_2 \in K_2| A = a, X_1 = x_1, X_2 = x_2 \big)\, P_{A | X_1, X_2}(da|x_1, x_2) \\
= & \int_{K_3} P\big(V_1 \in K_1, V_2 \in K_2| A = a\big)\, P_{A | X_1}(da|x_1).
\end{align*}
Thus, the conditional distribution of $\bU$ given $X_1, X_2$ does not depend on $X_2$, proving conditional independence.
\end{proof}

As already mentioned above, we consider now quantiles of differences and
other transformations of the dependent variables. To this end, let $\psi
(y_{1},y_{2})$ be an arbitrary (differentiable) function and note that 
\begin{equation}  \label{eq:genfuct}
\tilde{Y}=\psi (Y_{1},Y_{2})=\psi \left( \phi (X_{1},V_{1},A),\phi
(X_{2},V_{2},A)\right) =:g(X_{1},X_{2},{\boldsymbol{U}}),
\end{equation}%
so that for ${\boldsymbol{u}}=(v_{1},v_{2},a)$, we have that $g(x_{1},x_{2},{%
\boldsymbol{u}})=\psi (\phi (x_{1},v_{1},a),\phi (x_{2},v_{2},a))$. Denote
by $\tilde{q}(\tau ,x_{1},x_{2})$ the conditional quantile of $\tilde{Y}$
given $X=x_{1},X_{2}=x_{2}$, so that 
\begin{equation*}
P\left[ \tilde{Y}\leq \tilde{q}(\tau ,x_{1},x_{2})|X_{1}=x_{1},X_{2}=x_{2}%
\right] =\tau .
\end{equation*}%
For convenience, we first formulate and prove a result along the lines of
Hoderlein and Mammen (2007) for a general model of the form 
\begin{equation}
Y=g(X_{1},X_{2},{\boldsymbol{U}}),  \label{eq:simplifiedmodel}
\end{equation}%
in terms of regularity assumptions similar to Assumption \ref%
{ass:regquantile}, and then specialize it to (\ref{eq:genfuct}).

\begin{ass}
\label{ass:specialass} Suppose that in the model (\ref{eq:simplifiedmodel}),
we can write ${\boldsymbol{u}}=(h^{\prime },e)^{\prime }$ for scalar $e$ ,
such that $g(x_{1},x_{2},{\boldsymbol{u}})=g(x_{1},x_{2},h,e)$ is
continuously differentiable in $x_{2}$ and $e$. Moreover, for fixed $x_{1}$
there is a $C>0$ (possibly depending on $x_{1}$) with $\partial
_{e}g(x_{1},x_{2},h,e)\geq 1/C$ and $\left\Vert \partial
_{x_{2}}g(x_{1},x_{2},{\boldsymbol{u}})\right\Vert \leq C$ for all $x_{2}$
and ${\boldsymbol{u}}$. For the corresponding representation of the random
vector ${\boldsymbol{U}}=(H,E)$, $E$ is absolutely continuously distributed
given $(H,X_{1}),$ with conditional pdf $f_{E}(e|h,x_{1})$ that is bounded
and continuous in $e,$ and the conditional distribution of $H$ given $X_{1}$
is absolutely continuous with pdf $f(h|x_{1})$.
\end{ass}

These assumptions are by and large regularity conditions, akin to those
employed in Hoderlein and Mammen (2007), e.g., differentiability conditions.
They do not restrict the model significantly, and we therefore do not
discuss them at length. Together with the independence condition, they allow
us to establish an extension to the Hoderlein and Mammen (2007) result:

\begin{proposition}
\label{prop: dquant} Suppose that in the model (\ref{eq:simplifiedmodel}), $%
X_{2}$ is conditionally independent of ${\boldsymbol{U}}$ given $X_{1}$,
that Assumption \ref{ass:specialass} is satisfied and that the conditional
pdf of $Y$ given $X_{1}$ and $X_{2}$ is positive in the interior of its
support. Then for every $0<\tau <1$, the conditional quantile $q(\tau
,x_{1},x_{2})$ of $Y$ given $X_{1}=x_{1}$, $X_{2}=x_{2}$ exists and is
continuously differentiable with 
\begin{equation}  \label{eq:indentindepend}
\partial _{x_{2}}q(\tau ,x_{1},x_{2})=E\left[ \partial _{x_{2}}g(X_{1},X_{2},%
{\boldsymbol{U}})|Y=q(\tau ,x_{1},x_{2}),X_{1}=x_{1},X_{2}=x_{2}\right] .
\end{equation}
\end{proposition}

We now specialize this general result to the setup of this paper, and
discuss it below in this specialized setup. To this end, we modify the
regularity conditions accordingly:

\begin{ass}
\label{ass:quantileass22} Suppose that in the model (\ref{eq:simplifiedmodel}%
), $\psi (y_{1},y_{2})$ is continuously partially differentiable in $y_{2}$
with $1/K\leq \partial _{y_{2}}\psi (y_{1},y_{2})\leq K$ for all $%
y_{1},y_{2} $ for some $K>0$. Further, assume that we can write $v=(\tilde{h}%
^{\prime },e)^{\prime }$ for scalar $e$ , such that $\phi (x,v,a)=\phi (x,%
\tilde{h},e,a)$ is continuously differentiable in $x$ and $e$ and such that
there is a $C>0$ with $\partial _{e}\phi (x,h,e,a)\geq 1/C$ and $|\partial
_{x}\phi (x,v,a)|\leq C$ for all $x,v,a$. For the corresponding
representation of the random vector $V_{2}=(H,E)$, $E$ is absolutely
continuously distributed given $(X_{1},H,A,V_{1}),$ with conditional pdf
that is bounded and continuous in $e,$ and the conditional distribution of $%
(H,A,V_{1})$ given $X_{1}$ is absolutely continuous.
\end{ass}

These preliminaries lead to the expected corollary:

\begin{corollary}
Suppose that in (\ref{eq:genfuct}), Assumptions \ref{assum:model}, \ref%
{assum:condindependent} and \ref{ass:quantileass22} are satisfied, and that
the conditional density of $\tilde{Y}$ given $X_{1}=x_{1},X_{2}=x_{2}$ is
positive in the interior of its support. Then (\ref{eq:indentindepend})
holds true.
\end{corollary}

This result is very similar in spirit to the results in the previous
section, again an LAR for a subpopulation (or a derivative for an ASF) is
identified. The advantage, however, is now that we can look at
subpopulations that are characterized by arbitrary combinations of $Y_{1}$
and $Y_{2}$. If we confine ourselves to linear combinations, i.e., $\tilde{Y}%
=\lambda Y_{1}+\pi Y_{2},$ we can consider conditioning on arbitrary weights 
$\lambda ,\pi $. Since we can vary $\lambda ,\pi $ freely, this means that
we can use the entire joint distribution in the sense of the Cramer-Wold
device, by looking at any linear combination, and hence use multivariate
information through repeated use of one regular regression quantiles. It
allows to construct subpopulations where we put different weights on the
outcome in different periods. For instance, if $X$ is schooling, and $Y_{t}$
is labor income in different periods, we may think of $\tilde{Y}$ as some
long run or average income. And when computing this long run income, we
could either discount future income stronger or emphasize it more when
characterizing the subpopulations, depending on the intention of the
researcher. Of course, one should always remember that the strength in
statements we can make always comes at the expense of the structure we
impose on the dependence between $A$ and $X_{t}.$

\bigskip

This result covers important special cases:

\begin{enumerate}
\item The difference: $\psi (y_{1},y_{2})=y_{2}-y_{1}=\Delta y$. Then $%
\tilde{q}(\tau ,x_{1},x_{2})$ is the conditional quantile of the difference,
and 
\begin{equation*}
\partial _{x_{2}}q^{\Delta Y}(\tau ,x_{1},x_{2})=E\left[ \partial _{x}\phi
(x_{2},V_{2},A)|X_{1}=x_{1},X_{2}=x_{2},\Delta Y=q^{\Delta Y}(\tau
,x_{1},x_{2})\right] .
\end{equation*}

\item $Y_{2}$: Here $\psi (y_{1},y_{2})=y_{2}$, so that 
\begin{equation*}
\partial _{x_{2}}q^{Y_{2}}(\tau ,x_{1},x_{2})=E\big(\partial _{x}\phi
(x_{2},V_{2},A)|X_{1}=x_{1},X_{2}=x_{2},Y_{2}=q^{Y_{2}}(\tau ,x_{1},x_{2})%
\big).
\end{equation*}%
This is similar in spirit to Altonji and Matzkin (2005), just replacing
means by quantiles. 
\end{enumerate}

Note that the first special case answers one of the questions posed in the
introduction:\ should we consider the difference of the quantiles or the
quantiles of the differences, when talking about causal effects in panels.
In terms of the strength of the assumptions, the verdict has to be clearly
differences of quantiles. However, two remarks are in order: First, it also
happens to be the case that under the additional structure on the dependence
the quantiles of the difference yield a new effect that we could not have
obtained through differences in quantiles. In particular, for targeted
policy measures it may be sensible to use subpopulations that are defined
by, e.g., first differences $\Delta Y$. More precisely, since individuals
are often assumed to exhibit a pronounced loss aversion, i.e., they are more
much sensitive towards a negative change in their status than a positive, it
is conceivable that a policy maker would be much more interested in the
subpopulation for which the effect $\Delta Y$ is negative. Similarly,
measures that focus on the subpopulation exhibiting large values of $\Delta Y
$ may be of interest, as high variance of $Y$ over time may not be a
desirable feature for an individual. 

Second, with more time periods we could weaken the restrictive independence
assumptions. In particular, if three periods are available and only effects
on supopulations defined by, say, first differences between two periods are
of interest, we may allow for more correlation between the unobservables and
the $X_{t}$ process, and use the third period to perform an analogous
correction as in the previous section. Since this involves a simple
combination of arguments, we do not elaborate on this further, and we still
want to point to the difference in assumptions in the two periods case.

\section{Time Effects}

The time homogeneity assumption is a strong one that often seems not to hold
in applications. In this section we consider one way to weaken it, by
allowing for additive location effects and multiplicative scale effects.
Allowing for such time effects leads to effects of interest being exactly
identified, unlike the overidentification we found in Sections 2 and 3.

We allow for time effects by replacing Assumption 1 with the following
condition.

\bigskip

\begin{ass}
\label{ass:model2} There are functions $\phi,$ $\mu_t,$ and $\sigma_t,$ and
a vector of random variables $U_{t}$ such that 
\begin{equation*}
Y_{t}=\mu_t(X_{t})+ \sigma_t(X_{t}) \phi (X_{t}, U_{t}), \ \ (t = 1,2).
\end{equation*}
\end{ass}

The time effects $\mu_t$ and $\sigma_t$ are not separately identifiable from 
$\phi$ without location and scale normalizations because 
\begin{equation*}
\mu_t(x)+ \sigma_t(x) \phi (x, u) = \tilde \mu_t(x)+ \tilde \sigma_t(x)
\tilde \phi (x, u),
\end{equation*}
for $\tilde \mu_t(x) = \mu_t(x) + \sigma_t(x) \Delta_{\mu} (x),$ $\tilde
\sigma_t(x) = \Delta_{\sigma}(x )\sigma_t(x),$ $\tilde \phi(x,u) =
[\phi(x,u) - \Delta_{\mu}(x)] /\Delta_{\sigma}(x)$, and $\Delta_{\sigma}(x)
\neq 0$.

In this model the effects of interest vary with time. We consider the
time-averaged conditional mean effect: 
\begin{equation*}
\partial_x \bar \mu(x) + \partial_x \bar \sigma(x) E[\phi(x,U_t) \mid X_1 =
X_2 =x] + \bar \sigma(x) E[\partial_x \phi(x,U_t) \mid X_1 = X_2 =x],
\end{equation*}
and the time-averaged conditional quantile effect: 
\begin{equation*}
\partial_x \bar \mu(x) + \partial_x \bar \sigma(x) q(\tau,x) + \bar
\sigma(x) E[\partial_x \phi(x,U_t) \mid \phi(x,U_t) = q(\tau, x), X_1 = X_2
=x],
\end{equation*}
where $\bar \mu(x) = [\mu_1(x) + \mu_2(x)]/2,$ $\bar \sigma(x) =
[\sigma_1(x) + \sigma_2(x)]/2,$ and $q(\tau, x)$ is the $\tau^{th}$
conditional quantile of $\phi(x,U_t)$ given $X_1 = X_2 = x.$

The conditional mean effect is related to the time-averaged CASF: 
\begin{equation*}
\bar m(x \mid {\boldsymbol{x}}) = \bar \mu(x) + \bar \sigma(x) E[\phi(x,U_t)
\mid {\boldsymbol{X}} = {\boldsymbol{x}}],
\end{equation*}
through 
\begin{equation*}
\partial_x \bar m(x \mid {\boldsymbol{x}})\Big|_{{\boldsymbol{x}} = (x,x)} =
\partial_x \bar \mu(x) + \partial_x \bar \sigma(x) E[\phi(x,U_t) \mid X_1 =
X_2 =x] + \bar \sigma(x) E[\partial_x \phi(x,U_t) \mid X_1 = X_2 =x],
\end{equation*}
under the conditions that permit interchanging the derivative and
expectation. Similarly, the conditional quantile effect is related to the
time-averaged CQSF, $\bar q_{\tau}(x \mid {\boldsymbol{X}} = {\boldsymbol{x}}%
)$, that gives the $\tau$-quantile of $\bar \mu(x) + \bar \sigma(x)
\phi(x,U_t)$ conditional on ${\boldsymbol{X}}= {\boldsymbol{x}}$, through 
\begin{equation*}
\partial_x \bar q_{\tau}(x \mid {\boldsymbol{x}})\Big|_{{\boldsymbol{x}} =
(x,x)} =\partial_x \bar \mu(x) + \partial_x \bar \sigma(x) q(\tau,x) + \bar
\sigma(x) E[\partial_x \phi(x,U_t) \mid \phi(x,U_t) = q(\tau, x), X_1 = X_2
=x].
\end{equation*}

Let $V_t({\boldsymbol{x}}) =\text{Var}[Y_t \mid {\boldsymbol{X}} = {%
\boldsymbol{x}}],$ and $\sigma(x) = \sigma_2(x)/\sigma_1(x).$

\medskip

\begin{theorem}
\label{th:casf_te} Suppose that Assumptions \ref{ass:distrinv} and \ref%
{ass:model2} are satisfied, $E[Y_{t}^2]<\infty ,$ $(t=1,2),$ $V_t(x,x) > 0,$ 
$(t=1,2),$ $\phi (x,u), $ $\mu_t(x),$ and $\sigma_t(x),$ $(t=1,2),$ are
continuously differentiable in $x,$ and the conditional density of $U_{t}$
given ${\boldsymbol{X}}={\boldsymbol{x}},$ $f(u|{\boldsymbol{x}}),$ is
continuously differentiable in ${\boldsymbol{x}}$. Given $x$, suppose that
for some $\varepsilon >0$, 
\begin{align*}
& \int \sup_{\left\Vert \delta \right\Vert \leq \varepsilon, \delta =
(\delta_0^{\prime }, \delta_1^{\prime }, \delta_2^{\prime })^{\prime }}%
\big\Vert \partial_x \phi(x+\delta_0 ,u)\, f(u|x+\delta_{1},x+\delta _{2}) %
\big\Vert du<\infty , \\
& \int \sup_{\left\Vert \delta \right\Vert \leq \varepsilon, \delta =
(\delta_0^{\prime }, \delta_1^{\prime }, \delta_2^{\prime })^{\prime }}%
\big\Vert \phi(x+\delta_0 ,u)\, \partial_{x_s} f(u|x+\delta_{1},x+\delta
_{2}) \big\Vert du<\infty, \quad s=1,2.
\end{align*}
Then, $\sigma^2(x) = V_2(x,x)/V_1(x,x),$ $\mu_2(x) - \mu_1(x) \sigma(x) =
E[Y_2 - \sigma(x) Y_1 \mid X_1 = X_2 =x],$ \textit{and } 
\begin{multline*}
\partial_x \bar \mu(x) + \partial_x \bar \sigma(x) E[\phi(x,U_t) \mid X_1 =
X_2 =x] + \bar \sigma(x) E[\partial_x \phi(x,U_t) \mid X_1 = X_2 =x] \\
= [\partial_{x_1} M_1(x,x) - \partial_{x_1} M_2(x,x)/\sigma(x)]/2 +
[\partial_{x_2} M_2(x,x) - \sigma(x) \partial_{x_2} M_1(x,x)]/2.
\end{multline*}
\end{theorem}

\bigskip

This theorem shows that the time effects are identified up to location and
scale normalizations. For example, if we set $\mu_1(x) = 0$ and $\sigma_1(x)
=1,$ then $\sigma_2^2(x) = V_2(x,x)/V_1(x,x)$ and $\mu_2(x) = E[Y_2 -
\sigma_2(x) Y_1 \mid X_1 = X_2 =x]$. The identification of the conditional
mean effect does not require any normalization. Note that we now have just
one equation for identifying the conditional mean effect. 

\bigskip

We find a similar result for quantiles.

\begin{theorem}
\label{th:cqsf_te} Suppose that Assumptions \ref{ass:distrinv} , \ref%
{ass:regquantile} , and \ref{ass:model2} are satisfied, $\mu_t(x)$ and $%
\sigma_t(x)$ are continuously differentiable in $x$ and $\sigma_t(x) > 0,$ $%
(t=1,2),$ $f(u|{\boldsymbol{x}})$ is continuously differentiable in ${%
\boldsymbol{x}}$, 
\begin{equation}  \label{eq:dommain_te}
\int \sup_{\left\Vert \Delta _{{\boldsymbol{x}}}\right\Vert \leq
\delta}\left\Vert \partial_{\boldsymbol{x}} f(u|{\boldsymbol{x}}+\Delta _{{%
\boldsymbol{x}}})\right\Vert du<\infty ,
\end{equation}
and the conditional density of $Y_t$ given ${\boldsymbol{X}}$ is positive on
the interior of its support. Then for all $0<\tau <1,$ $Q_{t}(\tau | {%
\boldsymbol{x}})$ exists and is continuously differentiable at ${\boldsymbol{%
x}} = (x^{\prime },x^{\prime })^{\prime }$ such that 
\begin{multline*}
\partial_x \bar \mu(x) + \partial_x \bar \sigma(x) q(\tau,x) + \bar
\sigma(x) E[\partial_x \phi(x,U_t) \mid \phi(x,U_t) = q(\tau,x), X_1 = X_2
=x] \\
= [\partial_{x_1} Q_1(\tau \mid x,x) - \partial_{x_1} Q_2(\tau \mid
x,x)/\sigma(x)]/2 + [ \partial_{x_2} Q_2(\tau \mid x,x) - \sigma(x)
\partial_{x_2} Q_1(\tau \mid x,x)]/2,
\end{multline*}
$\sigma(x) = [ Q_2(\tau_1 \mid x,x) - Q_2(\tau_2 \mid x,x)]/[Q_1(\tau_1 \mid
x,x) - Q_1(\tau_2 \mid x,x)],$ and $\mu_2(x) - \sigma(x) \mu_1(x) =
Q_2(\tau_3 \mid x,x) - \sigma(x) Q_1(\tau_3 \mid x,x),$ for any $0 < \tau_3
< 1$ and $0 < \tau_2 < \tau_1 < 1$ such that $[Q_1(\tau_1 \mid x,x) -
Q_1(\tau_2 \mid x,x)] > 0$.
\end{theorem}

As in Theorem \ref{th:casf_te}, the time effects are identified up to
location and scale normalizations, whereas the conditional quantile effects
are identified without any normalization. Here, however, instead of
conditional mean and variance restrictions, we use quantile restrictions to
identify the time effects up to the normalizations. These effects are over
identified by many possible quantiles $\tau _{1},$ $\tau _{2}$ and $\tau
_{3} $. For example, for $\tau _{1}=.9,$ $\tau _{2}=.1$ and $\tau _{3}=.5,$
the scale is identified by a ratio of conditional interdecile ranges across
time and the location is identified by a difference of conditional medians
across time.

We note that Graham and Powell (2012) allowed for random time effects in
location and slope rather than location and scale effects that could depend
on $X$.

%
\bigskip

\section{Estimation and inference}

\label{sec:estimation}

The conditional mean and quantile effects of interest are identified by
special cases of the functionals: 
\begin{equation*}
\theta_m(x) = h_m(\{M_t(x,x), V_t(x,x) : t =1,2\}), \ \ x \in \mathcal{X},
\end{equation*}
and 
\begin{equation*}
\theta_q(w) = h_q(\{Q_t(\tau \mid x,x): t =1,2\}), \ \ w = (x,\tau) \in 
\mathcal{W},
\end{equation*}
respectively, where $h_m$ and $h_q$ are known smooth functions, $\mathcal{X}$
is a region of regressor values of interest, and $\mathcal{W}$ is a region
of regressor values and quantiles of interest. We consider the estimators of 
$\theta_m$ and $\theta_q$ based on the plug-in rule: 
\begin{equation*}
\widehat \theta_m(x) = h_m(\{\widehat M_t(x,x), \widehat V_t(x,x) : t
=1,2\}), \ \ x \in \mathcal{X},
\end{equation*}
and 
\begin{equation*}
\widehat \theta_q(w) = h_q(\{ \widehat Q_t(\tau \mid x,x) : t =1,2\}), \ \ w
= (x,\tau) \in \mathcal{W},
\end{equation*}
where $\widehat M_t(x,x)$, $\widehat Q_t(\tau \mid x,x)$, and $\widehat
V_t(x,x)$ are nonparametric series estimators of $M_t(x,x)$, $Q_t(\tau \mid
x,x)$, and $V_t(x,x)$.

To describe the series estimators, let $P^K({\boldsymbol{x}}) = (p_{1K}({%
\boldsymbol{x}}), \ldots, p_{KK}({\boldsymbol{x}}))^{\prime }$ denote a $K
\times 1$ vector of approximating functions, such as tensor products of
univariate polynomial or spline series terms of the components of $%
\boldsymbol{x}$, and let $\boldsymbol{P}_i = P^K({\boldsymbol{X}}_i)$. Then, 
\begin{equation*}
\widehat M_t(x,x) = P^K(x,x)^{\prime }\left( \sum_{i=1}^n \boldsymbol{P}_i 
\boldsymbol{P}_i^{\prime }\right)^- \sum_{i=1}^n \boldsymbol{P}_i Y_{it},
\end{equation*}
where $A^-$ denotes any generalized inverse inverse of the matrix $A$; 
\begin{equation*}
\widehat V_t(x,x) = P^K(x,x)^{\prime }\left( \sum_{i=1}^n \boldsymbol{P}_i 
\boldsymbol{P}_i^{\prime }\right)^- \sum_{i=1}^n \boldsymbol{P}_i [Y_{it} -
\widehat M_t({\boldsymbol{X}}_i)]^2
\end{equation*}
is a series version of the (kernel) conditional variance estimator of Fan
and Yao (1998); and $\widehat Q_t(\tau \mid x,x) = P^K(x,x)^{\prime
}\widehat \beta_t(\tau),$ where $\widehat \beta_t(\tau)$ is the Koenker and
Bassett (1978) quantile regression estimator 
\begin{equation*}
\widehat \beta_t(\tau) \in \arg\min_{b \in \mathbb{R}^{K}} \sum _{i=1}^{n}
[\tau - 1\{Y_{it} \leq \boldsymbol{P}_i^{\prime }b\}][Y_{it} -\boldsymbol{P}%
_i^{\prime }b].
\end{equation*}

%
%
%
%

Following  Praestgaard  and  Wellner (1993), Hahn (1995), and
Chamberlain and Imbens (2003), we use weighted bootstrap for inference.\footnote{See also Ma and Kosorok (2005)
and Chen and Pouzo (2009, 2013) for other applications of weighted
bootstrap; we are grateful to a referee for pointing out the latter
references.}   To describe this method, let $%
(w_{1},\ldots,w_{n})$ be an i.i.d. sequence of nonnegative random variables
from a distribution with mean and variance equal to one (e.g., the standard
exponential distribution), independent of the data. The weighted bootstrap
uses the components of $(w_{1},\ldots,w_{n})$ as random sampling weights in
the construction of the bootstrap version of the series estimators. Thus,
the bootstrap versions of $\widehat \theta_m(w)$ and $\widehat \theta_q(w)$
are 
\begin{equation*}
\widehat \theta_m^*(x) = h(\{\widehat M^*_t(x,x), \widehat V^*_t(x,x) : t
=1,2\}), \ \ x \in \mathcal{X},
\end{equation*}
and 
\begin{equation*}
\widehat \theta_q^*(w) = h(\{\widehat Q^*_t(\tau \mid x,x) : t =1,2\}), \ \
w = (x,\tau) \in \mathcal{W},
\end{equation*}
where 
\begin{equation*}
\widehat M^*_t(x,x) = P^K(x,x)^{\prime }\left( \sum_{i=1}^n w_i \boldsymbol{P%
}_i \boldsymbol{P}_i^{\prime }\right)^- \sum_{i=1}^n w_i \boldsymbol{P}_i
Y_{it}
\end{equation*}
is the bootstrap version of $\widehat M_t(x,x),$ 
\begin{equation*}
\widehat V^*_t(x,x) = P^K(x,x)^{\prime }\left( \sum_{i=1}^n w_i \boldsymbol{P%
}_i \boldsymbol{P}_i^{\prime }\right)^- \sum_{i=1}^n w_i \boldsymbol{P}_i
[Y_{it} - \widehat M^*_t({\boldsymbol{X}}_i)]^2
\end{equation*}
is the bootstrap version of $\widehat V^*_t(x,x),$ and $\widehat Q^*_t(\tau
\mid x,x) = P^K(x,x)^{\prime }\widehat \beta^*_t(\tau)$ is the bootstrap
version of $\widehat Q_t(\tau \mid x,x)$, with 
\begin{equation*}
\widehat \beta^*_t(\tau) = \arg\min_{b \in \mathbb{R}^{K}} \sum _{i=1}^{n}
w_i [\tau - 1\{Y_{it} \leq \boldsymbol{P}_i^{\prime }b\}][Y_{it} -%
\boldsymbol{P}_i^{\prime }b].
\end{equation*}

%
%

Belloni, Chernozhukov, Chetverikov, and Kato (2013) and Chernozhukov, Lee,
and Rosen (2013) developed functional distributional theory and bootstrap
consistency for series estimators of functionals of the conditional mean
function, and Belloni, Chernozhukov, and Fernandez-Val (2011) developed
similar theory for series estimators of functionals of the conditional
quantile function. We can use these results to construct analytical or
bootstrap confidence bands for the effects that have uniform asymptotic
coverage over regressor values and quantiles. For example, the end-point
functions of a $1-\alpha$ confidence band for $\theta_q$ have the form 
\begin{equation}  \label{eq: bands}
\widehat \theta_q^{\pm} (w) = \widehat \theta_q(w) \pm \widehat
t_{q,1-\alpha} \widehat \Sigma_q(w)^{1/2}/\sqrt{n},
\end{equation}
where $\widehat \Sigma_q(w)$ and $\widehat t_{q,1-\alpha}$ are consistent
estimators of the asymptotic variance function of $\sqrt{n} [\widehat{\theta}%
_q(w) - \theta_q(w)]$ and the $1 - \alpha$ quantile of the
Kolmogorov-Smirnov maximal $t$-statistic 
\begin{equation*}
t_q = \sup_{w \in {\mathcal{W}}} \widehat \Sigma_q(w)^{-1/2} \sqrt{n} |%
\widehat{\theta}_q(w) - \theta_q(w)|.
\end{equation*}

The following algorithm describes how to obtain uniform bands for quantile
effects using weighted bootstrap:

\begin{algorithm}[Uniform inference]
\label{alg:bands} \label{algorithm: inference} (i) Draw $\{\widehat{Z}%
_{q,b}^{\ast }:1\leq b\leq B\}$ as i.i.d. realizations of $\widehat{Z}%
_q^{\ast }(w)= \sqrt{n}[\widehat{\theta}_q^*(w)-\widehat{\theta}_q(w)] ,$
for $w\in {\mathcal{W}},$ conditional on the data. (ii) Compute a bootstrap
estimate of $\Sigma_q(w)^{1/2}$ such as the bootstrap standard deviation: $%
\widehat{\Sigma }_q(w)^{1/2}= \{\sum_{b=1}^B [\widehat{Z}_{q,b}^{\ast }(w) - 
\overline{Z}^{\ast } _{q}(w)]^2/(B-1)\}^{1/2}$ for $w\in {\mathcal{W}}$,
where $\overline{Z}^{\ast } _{q}(w) = \sum_{b=1}^B \widehat{Z}_{q,b}^{\ast
}(w) /B$; or the bootstrap interquartile range of $\widehat{Z}_q^{\ast }(w)$
rescaled with the normal distribution: $\widehat{\Sigma }_q(w)^{1/2}=[%
\widehat{Z}_{q,0.75}^{\ast }(w)-\widehat{Z}_{q,0.25}^{\ast }(w)]/1.349$ for $%
w\in {\mathcal{W}}$, where $\widehat{Z}_{q,p}^{\ast }(w)$ is the $p$-sample
quantile of $\{\widehat{Z}_{q,b}^{\ast }(w):1\leq b\leq B\} $. (3) Compute
realizations of the bootstrap version of the maximal t-statistic $\widehat{t}%
^*_{q,b}=\sup_{w\in {\mathcal{W}}} \widehat{\Sigma }_q(w)^{-1/2}|\widehat{Z}%
_{q,b}^{\ast }(w)|$ for $1 \leq b \leq B.$ (iii) Form a $(1-\alpha ) $%
-confidence band for $\{\theta(w)_q:w\in {\mathcal{W}} \}$ using (\ref{eq:
bands}) setting $\widehat{t}_{q,1-\alpha}$ to the $(1-\alpha ) $-sample
quantile of $\{\widehat{t}^*_{q,b}:1\leq b\leq B\}$. \qed
\end{algorithm}

The validity of Algorithm \ref{alg:bands} follows from the results in
Belloni, Chernozhukov, and Fernandez-Val (2011) and the delta method. We can
construct uniform bands for the conditional mean effects with a similar
algorithm replacing $\theta_q(w)$ by $\theta_m(x)$, adjusting all the steps
accordingly, and relying on the results of Belloni, Chernozhukov,
Chetverikov, and Kato (2013) and Chernozhukov, Lee, and Rosen (2013).

%

\section{Engel Curves in Panel Data}

\label{engel}

In this section, we illustrate the results with an empirical application on
estimation of Engel curves with panel data. The Engel curve relationship
describes how a household's demand for a commodity changes as the
household's expenditure increases. Lewbel (2006) provides a recent survey of
the extensive literature on Engel curve estimation. We use data from the
2007 and 2009 waves of the Panel Study of Income Dynamics (PSID). Since
2005, the PSID gathers information on household expenditure for different
categories of commodities. The PSID does not collect information on total
expenditure. We construct the total expenditure on nondurable goods and
services by adding all the expenses in housing, utilities, phone, child
care, food at home, food out from home, car, transportation, schooling,
clothing, leisure, and health. We exclude expenses in mortgage, home
insurance, car insurance, and health insurance because these categories have
many missing values. Our sample contains $968$ households formed by couples
without children, where the head of the household was 20 to 65 year-old in
2009, and that provided information about all the relevant categories of
expenditure in 2007 and 2009. We focus on the commodities food at home and
leisure for comparability with recent studies (e.g., Blundell, Chen, and
Kristensen (2007), Chen and Pouzo (2009, 2013), and Imbens and Newey
(2009)). The expenditure share on a commodity is constructed by dividing the
expenditure in this commodity by the total expenditure in nondurable goods
and services.

\medskip

Endogeneity in the estimation of Engel curves arises because the decision to
consume a commodity may occur simultaneously with the allocation of income
between consumption and savings. In contrast with the previous cross
sectional literature, we do not rely on a two-stage budgeting argument that
justifies the use of labor income as an instrument for expenditure. Instead,
we assume that the Engel curve relationships are time homogeneous up to
location and scale time effects, and rely on the availability of panel data.
Specifically, we estimate 
\begin{equation*}
Y_{it} = \mu_t(X_{it}) + \sigma_t(X_{it}) \phi(X_{it}, U_{it}), \ \ i = 1,
\ldots, 968, \ t = 1, 2,
\end{equation*}
where $Y$ is the observed share of total expenditure on food at home or
leisure, $X$ is the logarithm of total expenditure in dollars of 2005, $%
\mu_t(X)$ and $\sigma_t(X)$ are location and scale time effects, $U$ is a
vector of unobserved household heterogeneity that satisfies time homogeneity
and captures both differences in preferences and idiosyncratic household
shocks, $t=1$ corresponds to 2007, and $t=2$ corresponds to 2009.\footnote{%
To deflate total expenditure, we use a price index for personal consumption
expenditures in nondurable goods constructed from Tables 2.4.4 and 2.4.5 of
the Bureau of Economic Analysis.} The inclusion of time effects might be
important to account for temporal changes in preferences and relative prices
across commodities. For example, the price index of nondurable goods
increased by $7\%$ between 2007 and 2009, whereas the price indexes for food
and leisure increased by $10\%$ and $6\%$ during the same period.\footnote{%
Source: Tables 2.4.4.U and 2.4.5 of Bureau of Economic Analysis.} We allow
these time effects to vary with total expenditure, what gives flexibility to
the model. This model does put some restrictions on interactions between
prices and heterogeneity, implying that price changes only shift the
location and scale of the distribution of demand.

\medskip

Table \ref{table:desc_stat} reports descriptive statistics for the variables
used in the analysis. Both total expenditure and expenditure shares display
within and between household variation, with means and standard deviations
that remain stable between 2007 and 2009. The low percentage of within
variation in expenditure indicates that there might be a substantial number
of households with zero or little change in expenditure across years. Figure %
\ref{fig:density} plots histogram and kernel estimates of the density of the
change in expenditure between 2007 and 2009. The kernel estimates are
obtained using a Gaussian kernel with Silverman's rule of thumb for the
bandwidth. The estimates confirm that there is a high density of households
with zero change in expenditure. Our methods will identify mean and quantile
effects for these households with $X_{i1} = X_{i2}$.

\begin{table}[htp]
\caption{Descriptive Statistics}
\label{table:desc_stat}
\begin{center}
\begin{tabular}{lccccccc}
\hline\hline
\multicolumn{1}{c}{} & \multicolumn{3}{c}{Pooled sample} & 
\multicolumn{2}{c}{2007 sample} & \multicolumn{2}{c}{2009 sample} \\ 
\multicolumn{1}{c}{Variable} & Mean & Std. Dev. & Within (\%) & Mean & Std.
Dev. & Mean & Std. Dev. \\ \hline
Log expenditure & 10.13 & 0.56 & 21 & 10.19 & 0.57 & 10.07 & 0.55 \\ 
Food share & 0.19 & 0.10 & 28 & 0.19 & 0.10 & 0.20 & 0.10 \\ 
Leisure Share & 0.10 & 0.10 & 26 & 0.10 & 0.09 & 0.10 & 0.10 \\ \hline\hline
\multicolumn{8}{l}{{\footnotesize {Note: the source of the data is the PSID.
The number of observations is 968 observations for each year.}}}%
\end{tabular}%
\end{center}
\end{table}

\medskip

We estimate the location time effects, scale time effects, conditional mean
effects, and conditional quantile effects using sample analogs of the
expressions in Theorems \ref{th:casf_te} and \ref{th:cqsf_te}. In
particular, we estimate the conditional expectation, variance, and quantile
functions by the nonparametric series methods described in Section \ref%
{sec:estimation}. We consider two different specifications for the series
basis in all the estimators: a quadratic orthogonal polynomial and a cubic
B-spline with three knots at the minimum, median and maximum of total
log-expenditure in the data set. Both specifications are additively
separable in the total log-expenditures of 2007 and 2009.\footnote{%
We select these specifications by under smoothing with respect to the
specification selected by cross validation applied to the estimators of the
conditional expectation function.} 
We also compute cross sectional estimates that do not account for
endogeneity. They are obtained by averaging the nonparametric series
estimates in 2007 and 2009 that use the same specification of series basis
as the panel estimates but only condition on contemporaneous expenditure. 
For inference, we construct 90\% confidence bands around the estimates by
weighted bootstrap with exponential weights and 499 repetitions. These bands
are uniform in that they cover the entire function of interest with 90\%
probability asymptotically.

\medskip

Figures \ref{fig:te2_asf_food_in} and \ref{fig:te2_asf_leisure} show the
estimates and confidence bands for the time effects functions: 
\begin{equation*}
x \mapsto \mu_2(x) - \mu_1(x) \sigma(x), \ \ x \mapsto \sigma(x) =
\sigma_2(x)/\sigma_1(x), \ \ x \in \mathcal{X},
\end{equation*}
based on Theorem \ref{th:cqsf_te} with $\tau_1 = .9,$ $\tau_2 = .1,$ and $%
\tau_3 = .5$, where $\mathcal{X}$ is the interval of values between the 0.10
and 0.90 sample quantiles of log-total expenditure. We find that we cannot
reject the hypothesis that there are no location and scale time effects for
food at home, whereas we find significant evidence of time effects for
leisure with both series specifications. In results not reported, we find
similar estimates and confidence bands for the time effects functions based
on conditional means and variances using Theorem \ref{th:casf_te}.

\medskip


Figure \ref{fig:dqsf} plots the estimates and confidence bands for the
time-averaged conditional quantile effects or CQSF derivates integrated over
the values of $x$: 
\begin{equation*}
\tau \mapsto \int \left\{ \bar{\mu}(x) + \partial_x \bar{\sigma}(x)
q(\tau,x)+ \bar{\sigma}(x) E[\partial_x \phi(x,U_{it}) \mid \phi(x,U_{it}) =
q(\tau,x), X_{i1} = X_{i2} = x] \right\} \mu(dx),
\end{equation*}
for $\tau \in \mathcal{T},$ where $\mu$ is the empirical measure of
log-expenditure, and $\mathcal{T} = [0.1,0.9].$ Here we find heterogeneity
in the Engel curve relationship across the distribution. The pattern of the
effect is increasing with the quantile index for both food at home and
leisure, although the estimates are not sufficiently precise to distinguish
these patterns from sampling noise. The cross sectional estimates plotted in
dashed lines lie outside the confidence band for leisure, indicating
significant evidence of endogeneity. We do not find such evidence for food
at home.

\medskip

In figures \ref{fig:dqsfs_food_in} and \ref{fig:dqsfs_leisure}, we show that
the panel estimates of the CQSF as a function of expenditure are decreasing
for food at home and increasing for leisure at low values of expenditure.
Imbens and Newey (2009) and Chen and Pouzo (2009, 2013) found similar
patterns in their estimates of the QSF and the quantile Engel curves,
respectively. Figure \ref{fig:dasf} plots the estimates and confidence bands
for the time-averaged conditional mean effects or CASF derivatives: 
\begin{equation*}
x \mapsto \partial_x \bar{\mu}(x) + \partial_x \bar{\sigma}(x)
E[\phi(x,U_{it}) \mid X_{i1} = X_{i2} = x] + \bar{\sigma}(x) E[\partial_x
\phi(x,U_{it}) \mid X_{i1} = X_{i2} = x], \ \ x \in \mathcal{X}.
\end{equation*}
We also again evidence of endogeneity for leisure in the mean effects, but
not for food at home. As in Blundell, Chen and Kristensen (2007), the
conditional ASF is decreasing in expenditure for food at home, whereas it is
increasing for leisure. We find that the curve is convex for food at home
and concave for leisure. Note, however, that we should interpret the shape
of our panel estimates with caution because they formally correspond to
multiple conditional QSFs and ASFs as the conditioning set $X_1=X_2 =x$
changes with $x$ along the curve.


\medskip

Overall, the empirical results show that our panel estimates of the Engel
curves are similar to previous cross sectional estimates based on IV methods
to deal with endogeneity. Thus, the Engel curve relationship is decreasing
for food at home and increasing for leisure. Moreover, we find evidence of
the presence of time effects and endogeneity for leisure, but not for food
at home. These finding are consistent with consumer preferences where food
at home is a necessity good with little effect on the marginal allocation of
income between consumption and savings. Leisure, on the other hand, is a
superior good that affects the marginal allocation of income between
consumption and savings. The Engel curve relationship is stable over time
for food at home, whereas it is sensitive to changes over time in
preferences and relative prices for leisure.

\appendix

\section{Proof of Theorem \protect\ref{theorem:mean}}

It follows from the differentiability of $\phi (x,u)$ and $f( u|x_{1},x_{2})$
and the dominance condition that 
\begin{equation*}
\tilde{M}(x, x_1, x_2) : =\int \phi (x,u)f(u|x_{1},x_{2})du
\end{equation*}
is continuously differentiable in $(x^{\prime }, x_1^{\prime }, x_2^{\prime
})^{\prime }$ in a neighborhood of $(x^{\prime }, x^{\prime }, x^{\prime })$%
, and that the order of differentiation and integration can be interchanged.
Furthermore, by the structure of the model and Assumption \ref{ass:distrinv}%
, for ${\boldsymbol{x}} = (x_1^{\prime }, x_2^{\prime })^{\prime }$, 
\begin{equation*}
M_{t}({\boldsymbol{x}})=E\big(\phi(X_{t}, U_t)|X_{1}=x_{1},X_{2}=x_{2}\big)=%
\tilde{M}(x_{t},x_{1},x_{2}).
\end{equation*}
Therefore, $M_{t}({\boldsymbol{x}})$ is continuously differentiable in ${%
\boldsymbol{x}},$ $t=1,2,$ in a neighborhood of $(x^{\prime }, x^{\prime
})^{\prime }$, and for $s=1,2$, 
\begin{align}  \label{eq:hilfproofmean}
\begin{split}
\partial_{x_s} M_{t}({\boldsymbol{x}}) &=\big( 1(s=t)\partial_x \tilde{M}(x,
x_1, x_2)+ \partial_{x_s} \tilde{M}(x, x_1, x_2)\big)\big|_{(x, x_1,
x_2)=(x_{t},x_{1},x_{2})} \\
&=\int \big( 1(s=t)\partial_x \phi (x_{t},u) f(u|{\boldsymbol{x}})+\phi
(x_{t},u) \partial_{x_s} f(u|{\boldsymbol{x}})\big) du \\
&=E\Big(\big(1(s=t)\partial_x \phi(x_{t},U_{t})+ \phi (x_{t},U_{t})h_s(U_{t}|%
{\boldsymbol{x}})\big) \big|{\boldsymbol{X}}={\boldsymbol{x}}\Big),
\end{split}%
\end{align}%
where $h_{s}(u|x)=f(u|{\boldsymbol{x}})^{-1}\partial_{x_s} f(u|{\boldsymbol{x%
}})$. Subtracting and using $U_{1}|{\boldsymbol{X}}\overset{d}{=}U_{2}|{%
\boldsymbol{X}}$, 
\begin{align*}
\partial_{x_2} M_{2}({\boldsymbol{x}}) - \partial_{x_2} M_{1}({\boldsymbol{x}%
}) & = E\big(\partial_x \phi(x_{2},U_{2})\big|{\boldsymbol{X}}={\boldsymbol{x%
}}\big) \\
& + E \Big( \big(\phi (x_{2},U_{2}) - \phi (x_{1},U_{2}) \big)\, h_2(U_{2}|{%
\boldsymbol{x}})\big) \big|{\boldsymbol{X}}={\boldsymbol{x}}\Big).
\end{align*}
Evaluating at ${\boldsymbol{x}}= (x^{\prime }, x^{\prime })^{\prime }$ gives
(\ref{eq:resultmean1}), and (\ref{eq:resultmean2}) follows similarly by
considering $\partial_{x_1} M_{1}({\boldsymbol{x}}) - \partial_{x_1} M_{2}({%
\boldsymbol{x}})$, using (\ref{eq:hilfproofmean}) and evaluating at ${%
\boldsymbol{x}}= (x^{\prime }, x^{\prime })^{\prime }$. \qed

\section{Proof of Lemma \protect\ref{lemm:quantilemain}}

Let $F_E(e |h ,{\boldsymbol{x}})=\Pr (E \leq e |H = h ,{\boldsymbol{X}}={%
\boldsymbol{x}})=\int_{-\infty }^{e} f_E(\epsilon|h ,{\boldsymbol{x}})\,
d\epsilon.$ Then by the fundamental theorem of calculus, $F_E(e|h ,{%
\boldsymbol{x}})$ is differentiable in $e$ with derivative $f_E(e|h ,{%
\boldsymbol{x}})$ that is continuous in $e $ and ${\boldsymbol{x}}$.
Consider 
\begin{align}  \label{eq:hilflemone}
\begin{split}
& \Pr (\phi (x, U_t) \leq y|{\boldsymbol{X}}={\boldsymbol{x}}) =\int 1(\phi
(x,u)\leq y)\,f(u|{\boldsymbol{x}})\,du \\
& =\int \int 1(\phi (x,h, e)\leq y)\, f_E(e| h, {\boldsymbol{x}})\, f(h |{%
\boldsymbol{x}})d e\, dh \\
& =\int \int 1(e \leq \phi^{-1} (x,h, y))\, f_E(e| h, {\boldsymbol{x}})\,
f(h |{\boldsymbol{x}})d e\, dh \\
& =\int F_E(\phi ^{-1}(x,h ,y)|h ,{\boldsymbol{x}})f(h |{\boldsymbol{x}}%
)\,dh .
\end{split}%
\end{align}
By the inverse and implicit function theorems, $\phi^{-1}(x,h ,y)$ is
continuously differentiable in $x$ and $y,$ with 
\begin{eqnarray*}
\partial_y \phi ^{-1}(x,h ,y) &=&\big[\partial_e \phi (x,h ,\phi ^{-1}(x,h
,y)) \big]^{-1}, \\
\partial_x \phi^{-1}(x,h ,y) &=&-\frac{\partial_x \phi (x,h ,\phi^{-1}(x,h
,y))}{\partial_e \phi (x,h ,\phi^{-1}(x,h ,y))} \\
& = & - \partial_x \phi (x,h ,\phi^{-1}(x,h ,y))\, \partial_y \phi ^{-1}(x,h
,y).
\end{eqnarray*}
Then by Assumption \ref{ass:regquantile} both $\partial_y \phi^{-1}(x,h ,y)$
and $\partial_x \phi^{-1}(x,h ,y)$ are continuous in $y$ and $x$ and
bounded. Therefore, 
\begin{align}  \label{eq:hilflemtwo}
\begin{split}
\partial_y F_E (\phi^{-1}(x,h ,y)|h ,{\boldsymbol{x}}) &=f_E(\phi^{-1}(x,h
,y)|h ,{\boldsymbol{x}})\, \partial_y \phi ^{-1}(x,h ,y) \\
& =f_{Y_{x}|{\boldsymbol{X}},H_t }(y|{\boldsymbol{x}},h ), \\
\partial_x F_E(\phi^{-1}(x,h ,y)|h ,{\boldsymbol{x}}) &=f_E(\phi^{-1}(x,h
,y)|h ,{\boldsymbol{x}})\, \partial_x \phi^{-1}(x,h ,y) \\
&=-f_{Y_{x}|{\boldsymbol{X}},H_t}(y|{\boldsymbol{x}},h ) \, \partial_x \phi
(x,h ,\phi^{-1}(x,h ,y)),
\end{split}%
\end{align}
are both bounded and continuous in $y,x$ and ${\boldsymbol{x}},$ where the
last equality in each equation follows by a standard change of variables
argument. From the boundedness assumptions on $f_E$ and on $\partial_x \phi$
in Assumption \ref{ass:regquantile}, it follows that $\Pr (\phi (x,U_t)\leq
y|{\boldsymbol{X}}={\boldsymbol{x}})$ is partially differentiable in $y$ and 
$x$ with partial derivatives continuous in $y,x,{\boldsymbol{x}}$, which can
be computed by differentiating under the integral in (\ref{eq:hilflemone}).
In order to establish the expressions in the lemma, insert (\ref%
{eq:hilflemtwo}) into the partial derivatives of (\ref{eq:hilflemone})
w.r.t. $y$ and $x$, and note that $f_{Y_{x}|{\boldsymbol{X}},H_t}(y|{%
\boldsymbol{x}},h ) f(h |{\boldsymbol{x}}) = f_{Y_{x}, H_t|{\boldsymbol{X}}%
}(y,h |{\boldsymbol{x}})$. The first expression is then immediate. For the
second, note that given $Y_x=y$ (for a fixed $x$), $E_t = \phi^{-1}(x, H_t,
y)$, so that 
\begin{align*}
& f_{Y_{x}|{\boldsymbol{X}}}(y|{\boldsymbol{x}})\,E\big(\partial_x \phi
(x,U_t)|Y_{x}=y,{\boldsymbol{X}}={\boldsymbol{x}}\big) \\
= & f_{Y_{x}|{\boldsymbol{X}}}(y|{\boldsymbol{x}})\,\int \partial_x \phi(x,
h, \phi^{-1}(x,h,y))\, f_{H_t | Y_{x}, {\boldsymbol{X}}}(h| y,{\boldsymbol{x}%
}) \, dh \\
= & \int \partial_x \phi(x, h, \phi^{-1}(x,h,y))\, f_{Y_{x}, H_t | {%
\boldsymbol{X}}}(y, h|{\boldsymbol{x}}) \, dh.
\end{align*}
\qed

\section{Proof of Theorem \protect\ref{theorem:quant}}

Let ${\boldsymbol{z}}=(y,x^{\prime },{\boldsymbol{x}}^{\prime })^{\prime }$
and let $H({\boldsymbol{z}})=\Pr (\phi (x,U_{t})\leq y|{\boldsymbol{X}}={%
\boldsymbol{x}})$. From Lemma \ref{lemm:quantilemain} it follows that $H({%
\boldsymbol{z}})$ is differentiable in $y$ and $x$ with partial derivatives
continuous in ${\boldsymbol{z}}$.

From (\ref{eq:dommain}), it follows that $H({\boldsymbol{z}})$ is also
differentiable in ${\boldsymbol{x}}$ with 
\begin{equation*}
\partial_{\boldsymbol{x}} H({\boldsymbol{z}})=\int 1(\phi (x,u)\leq y)
\partial_{\boldsymbol{x}} f(u|{\boldsymbol{x}}) du
\end{equation*}
which is continuous in ${\boldsymbol{z}}$. Thus, $H({\boldsymbol{z}})$ is
continuously differentiable in ${\boldsymbol{z}}$, and the derivative w.r.t.~%
$y$ is strictly positive (see the expression in Lemma \ref{lemm:quantilemain}%
). From the implicit function theorem, there is a unique solution $%
Q_{t}(\tau | {\boldsymbol{x}})$, ${\boldsymbol{x}} = (x_1^{\prime },
x_2^{\prime })^{\prime }$, to 
\begin{equation*}
\tau =H(Q_{t}(\tau | {\boldsymbol{x}}),x_{t},{\boldsymbol{x}}), \qquad t=1,2.
\end{equation*}
which is differentiable with partial derivatives 
\begin{eqnarray*}
\partial_{x_s} Q_{t}(\tau | {\boldsymbol{x}}) =- \big(\partial_y
H(Q_{t}(\tau | {\boldsymbol{x}}),x_{t},{\boldsymbol{x}})\big)^{-1}\, \Big( %
1(s=t) \partial_x H(Q_{t}(\tau | {\boldsymbol{x}}),x_{t},{\boldsymbol{x}})+
\partial_{x_s} H(Q_{t}(\tau | {\boldsymbol{x}}),x_{t},{\boldsymbol{x}})\Big),
\end{eqnarray*}
where $\partial_{x_s} H\big(y,x,{\boldsymbol{x}})$ is the partial derivative
w.r.t.~the components of $x_s$ in ${\boldsymbol{x}} = (x_1^{\prime },
x_2^{\prime })^{\prime }$, $s=1,2$. Evaluating at $x_1 = x_2 = x$,
subtracting and plugging in the expressions for the derivatives from Lemma %
\ref{lemm:quantilemain} yields (\ref{eq:quantileeffect}). \qed

\section{Proof of Proposition \protect\ref{prop: dquant}}

Fix $x_1$, and set 
\begin{equation*}
H(y, x_2) = P(Y \leq y | X_1 = x_1, X_2 = x_2) = P(g(x_1, x_2, {\boldsymbol{U%
}}) \leq y| X_1 = x_1)
\end{equation*}
by the form of the model and the conditional independence assumption. Below
we show that from Assumption \ref{ass:specialass}, $H(y, x_2)$ is
continuously partially differentiable with derivatives 
\begin{align}  \label{hilf:generalagain}
\begin{split}
\partial_y H( y, x_2) &=f_{Y|X_1, X_2}(y|x_1, x_2), \\
\partial_{x_2} H(y, x_2) &= - f_{Y|X_1, X_2}(y|x_1, x_2)\, E\big(%
\partial_{x_2} g(X_1, X_2,{\boldsymbol{U}}) |Y=y,X_1 = x_1, X_2 = x_2\big).
\end{split}%
\end{align}
From the positivity of the conditional density of $Y$ given $X_1, X_2$ and
the implicit function theorem, the conditional quantile $q(\tau, x_1, x_2)$
given by $H\big(q(\tau, x_1, x_2), x_2 \big) = \tau$ exists and is
differentiable in $x_2$ with derivative 
\begin{align*}
\partial_{x_2} q(\tau, x_1, x_2) & = - \Big(\partial_y H\big( q(\tau, x_1,
x_2), x_2\big) \Big)^{-1} \partial_{x_2} H\big(q(\tau, x_1, x_2), x_2\big) \\
& = E\big(\partial_{x_2} g(X_1, X_2,{\boldsymbol{U}}) |Y=q(\tau, x_1,
x_2),X_1 = x_1, X_2 = x_2\big).
\end{align*}
where we used (\ref{hilf:generalagain}), thus proving the theorem.

It remains to prove (\ref{hilf:generalagain}), which is analogous to Lemma %
\ref{lemm:quantilemain}. Let $F_E(e |h ,{\boldsymbol{x}})=\Pr (E \leq e |H =
h ,X_1 = x_1)=\int_{-\infty }^{e} f_E(e|h ,x_1)\, de$, so that $F_E(e|h
,x_1) $ is differentiable in $e$ with derivative $f_E(e|h ,x_1)$ that is
continuous in $e $. Consider 
\begin{align}  \label{eq:hilflemone2}
\begin{split}
& \Pr (g(x_1, x_2, {\boldsymbol{U}}) \leq y|X_1 = x_1) =\int 1(g(x_1, x_2, {%
\boldsymbol{u}})\leq y)\,f({\boldsymbol{u}}|x_1)\,du \\
& =\int \int 1(g(x_1, x_2, h, e)\leq y)\, f_E(e| h, x_1)\, f(h |x_1)d e\, dh
\\
& =\int \int 1(e \leq g^{-1} (x_1, x_2, h, y))\, f_E(e| h, x_1)\, f(h |x_1)d
e\, dh \\
& =\int F_E(g^{-1} (x_1, x_2, h, y)|h ,x_1)f(h |x_1)\,dh .
\end{split}%
\end{align}
By the inverse and implicit function theorems, $g^{-1} (x_1, x_2, h, y)$ is
continuously differentiable in $x$ and $y,$ with 
\begin{eqnarray*}
\partial_y \big(g^{-1} (x_1, x_2, h, y)\big) &=&\big(\partial_e g(x_1, x_2,
h, g^{-1} (x_1, x_2, h, y)) \big)^{-1}, \\
\partial_{x_2} \big(g^{-1} (x_1, x_2, h, y)\big) &=&-\frac{\partial_{x_2}
g(x_1, x_2, h, g^{-1} (x_1, x_2, h, y))}{\partial_e g(x_1, x_2, h, g^{-1}
(x_1, x_2, h, y))} \\
& = & - \partial_{x_2} g(x_1, x_2, h, g^{-1} (x_1, x_2, h, y))\, \partial_y %
\big(g^{-1} (x_1, x_2, h, y)\big).
\end{eqnarray*}
Then by Assumption \ref{ass:specialass} both $\partial_y \big(g^{-1} (x_1,
x_2, h, y)\big) $ and $\partial_{x_2} \big(g^{-1} (x_1, x_2, h, y)\big)$ are
continuous in $y$ and $x_2$ and bounded. Therefore, 
\begin{align}  \label{eq:hilflemtwo2}
\begin{split}
\partial_y \big(F_E (g^{-1} (x_1, x_2, h, y)|h ,x_1)\big) &=f_E(g^{-1} (x_1,
x_2, h, y)|h ,x_1)\, \partial_y \big(g^{-1} (x_1, x_2, h, y)\big) \\
& =f_{Y|X_1, X_2,H }(y|x_1, x_2,h ), \\
\partial_{x_2} \big(F_E(g^{-1} (x_1, x_2, h, y)|h ,{\boldsymbol{x}})\big) %
&=f_E(g^{-1} (x_1, x_2, h, y)|h ,{\boldsymbol{x}})\, \partial_{x_2} \big(%
g^{-1} (x_1, x_2, h, y)\big) \\
&=-f_{Y|X_1, X_2,H }(y|x_1, x_2,h )\, \partial_{x_2} g(x_1, x_2, h, g^{-1}
(x_1, x_2, h, y)),
\end{split}%
\end{align}
are both bounded and continuous in $y,x$ and ${\boldsymbol{x}},$ where the
last equality in each equation follows by a change of variables argument
together with conditional independence of $X_2$ and $(H, E)$ given $X_1$.
From the boundedness assumptions on $f_E$ and on $\partial_{x_2} g$ in
Assumption \ref{ass:specialass}, it follows that $H(y, x_2)$ is partially
differentiable with continuous partial derivatives which can be computed by
differentiating under the integral in (\ref{eq:hilflemone2}). Now insert (%
\ref{eq:hilflemtwo2}) into (\ref{eq:hilflemone2}) and note that $f_{Y|X_1,
X_2,H}(y|x_1, x_2,h ) f(h |x_1) = f_{Y, H|X_1, X_2}(y,h |x_1, x_2)$ by
conditional independence. The first expression is then immediate. For the
second, note that given $Y=y$ (for a fixed $x_1, x_2$), $E = g^{-1}(x_1,
x_2, H, y)$, so that 
\begin{align*}
& f_{Y|X_1, X_2}(y|x_1, x_2)\,E\big(\partial_{x_2} g(x_1, x_2,{\boldsymbol{U}%
})|Y=y,X_1 = x_1, X_2 = x_2\big) \\
= & f_{Y|X_1, X_2}(y|x_1, x_2)\,\int \partial_{x_2} g(x_1, x_2, h,
g^{-1}(x_1, x_2,h,y))\, f_{H | Y, X_1, X_2}(h| y,x_1, x_2) \, dh \\
= & \int \partial_{x_2} g(x_1, x_2, h, g^{-1}(x_1, x_2,h,y))\, f_{Y, H |
X_1, X_2}(y, h|x_1, x_2) \, dh.
\end{align*}
\qed

\section{Proof of Theorem \protect\ref{th:casf_te}}

The first result follows by direct calculation because 
\begin{equation*}
V_t(x,x) = \sigma_t^2(x) \text{Var}[\phi(x,U_t) \mid X_1 = X_2 = x],
\end{equation*}
and $\text{Var}[\phi(x,U_t) \mid X_1 = X_2 = x]$ does not depend on $t$ by
Assumption \ref{ass:distrinv}.

For the second result, note that 
\begin{equation*}
E[Y_2 - \sigma(x) Y_1 \mid X_1 = X_2 =x] = M_2(x,x) - \sigma(x) M_{1}(x,x).
\end{equation*}
Then the result follows by direct calculation because 
\begin{equation*}
M_t(x,x) = \mu_t(x) + \sigma_t(x) E[\phi(x,U_t) \mid X_1 = X_2 = x], \ \ 
\end{equation*}
and $E[\phi(x,U_t) \mid X_1 = X_2 = x]$ does not depend on $t$ by Assumption %
\ref{ass:distrinv}.

The proof of the third result is similar to the proof of Theorem \ref%
{theorem:mean} replacing $\phi(x,u)$ by $\phi_t(x,u) = \mu_t(x) +
\sigma_t(x) \phi(x,u)$. In particular, ${\boldsymbol{x}} \mapsto M_t({%
\boldsymbol{x}})$ is continuously differentiable, $t=1,2$, in a neighborhood
of $(x^{\prime },x^{\prime })^{\prime },$ and for $s=1,2,$ 
\begin{equation*}
\partial_{x_s} M_t({\boldsymbol{x}}) = E\Big(1(s=t)\partial_x
\phi_t(x_{t},U_{t})+ \sigma_t(x_t) \phi (x_{t},U_{t})h_s(U_{t}|{\boldsymbol{x%
}}) \big|{\boldsymbol{X}}={\boldsymbol{x}}\Big),
\end{equation*}
where $h_{s}(u|{\boldsymbol{x}})=f(u|{\boldsymbol{x}})^{-1}\partial_{x_s}
f(u|{\boldsymbol{x}})$ and $\partial_x \phi_t(x,u) = \partial_x \mu_t(x) +
\partial_x \sigma_t(x) \phi(x,u) + \sigma_t(x) \partial_x \phi(x,u)$.
Subtracting and using Assumption \ref{ass:distrinv}, 
\begin{align*}
\partial_{x_2} M_{2}({\boldsymbol{x}}) - \sigma(x) \partial_{x_2} M_{1}({%
\boldsymbol{x}}) & = E\big(\partial_x \phi_t(x_{2},U_{2})\big|{\boldsymbol{X}%
}={\boldsymbol{x}}\big) \\
& + E \Big( \big(\sigma_2(x_2) \phi (x_{2},U_{2}) - \sigma(x) \sigma_1(x_1)
\phi (x_{1},U_{2}) \big)\, h_2(U_{2}|{\boldsymbol{x}})\big) \big|{%
\boldsymbol{X}}={\boldsymbol{x}}\Big).
\end{align*}
Evaluating at ${\boldsymbol{x}}= (x^{\prime }, x^{\prime })^{\prime }$ gives 
\begin{equation*}
\partial_{x_2} M_{2}(x,x) - \sigma(x) \partial_{x_2} M_{1}(x,x) = E\big(%
\partial_x \phi_2(x,U_{2})\big| X_1 = X_2 = x \big).
\end{equation*}
A similar argument yields 
\begin{equation*}
\partial_{x_1} M_{1}(x,x) - \partial_{x_1} M_{2}(x,x)/\sigma(x) = E\big(%
\partial_x \phi_1(x,U_{1})\big| X_1 = X_2 = x \big).
\end{equation*}
The result follows by averaging the previous expressions and using that $E%
\big(\partial_x \phi(x,U_{t})\big| X_1 = X_2 = x \big)$ does not depend on $%
t $ by Assumption \ref{ass:distrinv}. \qed

\section{Proof of Theorem \protect\ref{th:cqsf_te}}

Let ${\boldsymbol{z}}=(y,x^{\prime },{\boldsymbol{x}}^{\prime })^{\prime }$
and let $H_t({\boldsymbol{z}})=\Pr (\phi_t (x,U_{t})\leq y|{\boldsymbol{X}}={%
\boldsymbol{x}})$ where $\phi_t (x,u) = \mu_t(x) + \sigma_t(x) \phi(x,u)$.
The first result follows by a similar argument to the proof of Theorem \ref%
{theorem:quant}. In particular, by Lemma \ref{lemm:quantilemain} and (\ref%
{eq:dommain_te}), ${\boldsymbol{z}} \mapsto H_t({\boldsymbol{z}})$ is
continuously differentiable with derivatives 
\begin{align*}
\partial_y H_t({\boldsymbol{z}}) &=f_{\phi_t(x,U_t)|{\boldsymbol{X}}}(y|{%
\boldsymbol{x}}), \\
\partial_x H_t({\boldsymbol{z}}) &= - f_{\phi_t(x,U_t)|{\boldsymbol{X}}}(y|{%
\boldsymbol{x}}) E\big(\partial_x \phi_t (x,U_t) |\phi_t(x,U_t)=y,{%
\boldsymbol{X}}={\boldsymbol{x}}\big), \\
\partial_{\boldsymbol{x}} H({\boldsymbol{z}}) &= \int 1(\phi_t (x,u)\leq y)
\partial_{\boldsymbol{x}} f(u|{\boldsymbol{x}}) du.
\end{align*}
Thus, ${\boldsymbol{z}} \mapsto H_t({\boldsymbol{z}})$ is continuously
differentiable with positive derivative with respect to $y$. By the implicit
function theorem, there is a unique solution $Q_{t}(\tau|{\boldsymbol{x}})$
to $\tau = H_t(Q_t(\tau|{\boldsymbol{x}}), x_t, {\boldsymbol{x}}),$ $t=1,2,$
which is differentiable with partial derivatives 
\begin{equation*}
\partial_{x_s} Q_{t}(\tau | {\boldsymbol{x}}) =- \big(\partial_y
H_t(Q_{t}(\tau | {\boldsymbol{x}}),x_{t},{\boldsymbol{x}})\big)^{-1}\, \Big( %
1(s=t) \partial_x H_t(Q_{t}(\tau | {\boldsymbol{x}}),x_{t},{\boldsymbol{x}}%
)+ \partial_{x_s} H_t(Q_{t}(\tau | {\boldsymbol{x}}),x_{t},{\boldsymbol{x}})%
\Big),
\end{equation*}
where $\partial_{x_s} H_t\big(y,x,{\boldsymbol{x}})$ is the partial
derivative w.r.t.~the components of $x_s$ in ${\boldsymbol{x}} =
(x_1^{\prime }, x_2^{\prime })^{\prime }$, $s=1,2$.

Evaluating at $x_1 = x_2 = x$, and plugging in the expressions for the
derivatives yields 
\begin{multline*}
\partial_{x_s} Q_{t}(\tau | {\boldsymbol{x}}) = 1(s=t)E\big(\partial_x
\phi_t (x,U_t) |\phi(x,U_t)=q(\tau,x),X_1=X_2 = x\big) \\
- [f_{\phi(x,U_t)|{\boldsymbol{X}}}(q(\tau,x)|{\boldsymbol{x}}%
)/\sigma_t(x)]^{-1} \int 1(\phi (x,u)\leq q(\tau,x)) \partial_{x_s} f(u|{%
\boldsymbol{x}}) du ,
\end{multline*}
where we use that $Q_{t}(\tau | {\boldsymbol{x}}) = \mu_t(x) + \sigma_t(x)
q(\tau,x)$ by invariance of quantiles to monotone transformations, and $%
f_{\phi_t(x,U_t)|{\boldsymbol{X}}}(y|{\boldsymbol{x}}) = f_{\phi(x,U_t)|{%
\boldsymbol{X}}}([y - \mu_t(x)]/\sigma_t(x) |{\boldsymbol{x}}) / \sigma_t(x)$
by a change of variables. Subtracting and using Assumption \ref{ass:distrinv}
\begin{equation*}
\partial_{x_1} Q_{1}(\tau | {\boldsymbol{x}}) - \sigma(x)^{-1}
\partial_{x_1} Q_{2}(\tau | {\boldsymbol{x}}) = E\big(\partial_x \phi_1
(x,U_1) |\phi(x,U_1)=q(\tau,x),X_1=X_2 = x\big),
\end{equation*}
and 
\begin{equation*}
\partial_{x_2} Q_{2}(\tau | {\boldsymbol{x}}) - \sigma(x) \partial_{x_2}
Q_{1}(\tau | {\boldsymbol{x}}) = E\big(\partial_x \phi_2 (x,U_2)
|\phi(x,U_2)=q(\tau,x),X_1=X_2 = x\big).
\end{equation*}
The result then follows by averaging the previous expressions, using $%
\partial_x \phi_t(x,u) = \partial_x \mu_t(x) + \partial_x \sigma_t(x)
\phi(x,u) + \sigma_t(x) \partial_x \phi(x,u)$ and that $E\big(\partial_x
\phi (x,U_t) |\phi(x,U_t)=q(\tau,x),X_1=X_2 = x\big)$ does not depend on $t$
by Assumption \ref{ass:distrinv}.

The second and third results follow from Assumption \ref{ass:distrinv} by
direct calculation because 
\begin{equation*}
Q_t( \tau \mid x,x) = \mu_t(x) + \sigma_t(x) q(\tau,x).
\end{equation*}
\qed

\newpage


\begin{figure}

\begin{center}

\centering\epsfig{figure=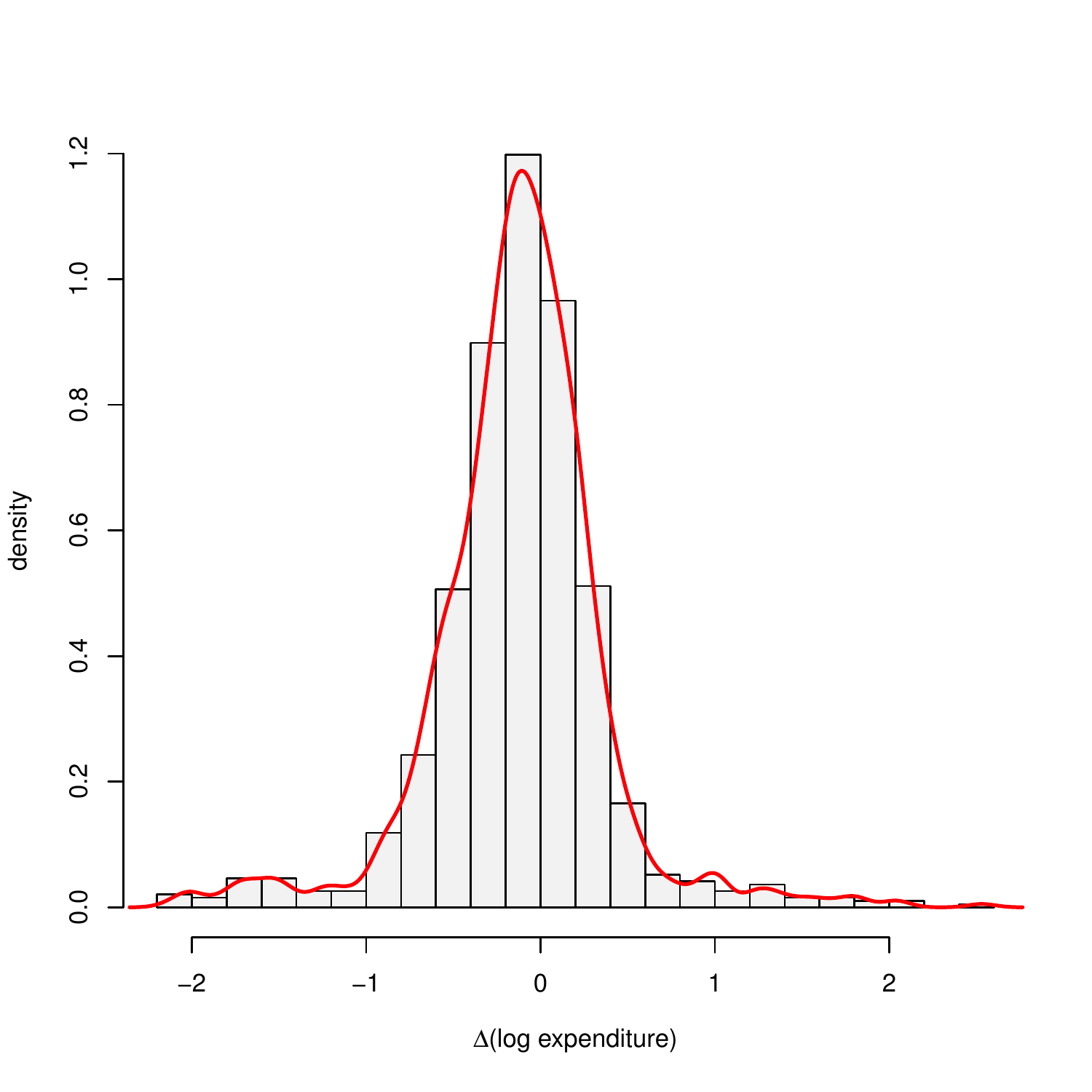,width=6.5in,height=6.5in,angle=0}

\caption{\label{fig:density} Density of the change in log-expenditure between 2007 and 2009.}

\end{center}

\end{figure}



\begin{figure}

\begin{center}

\centering\epsfig{figure=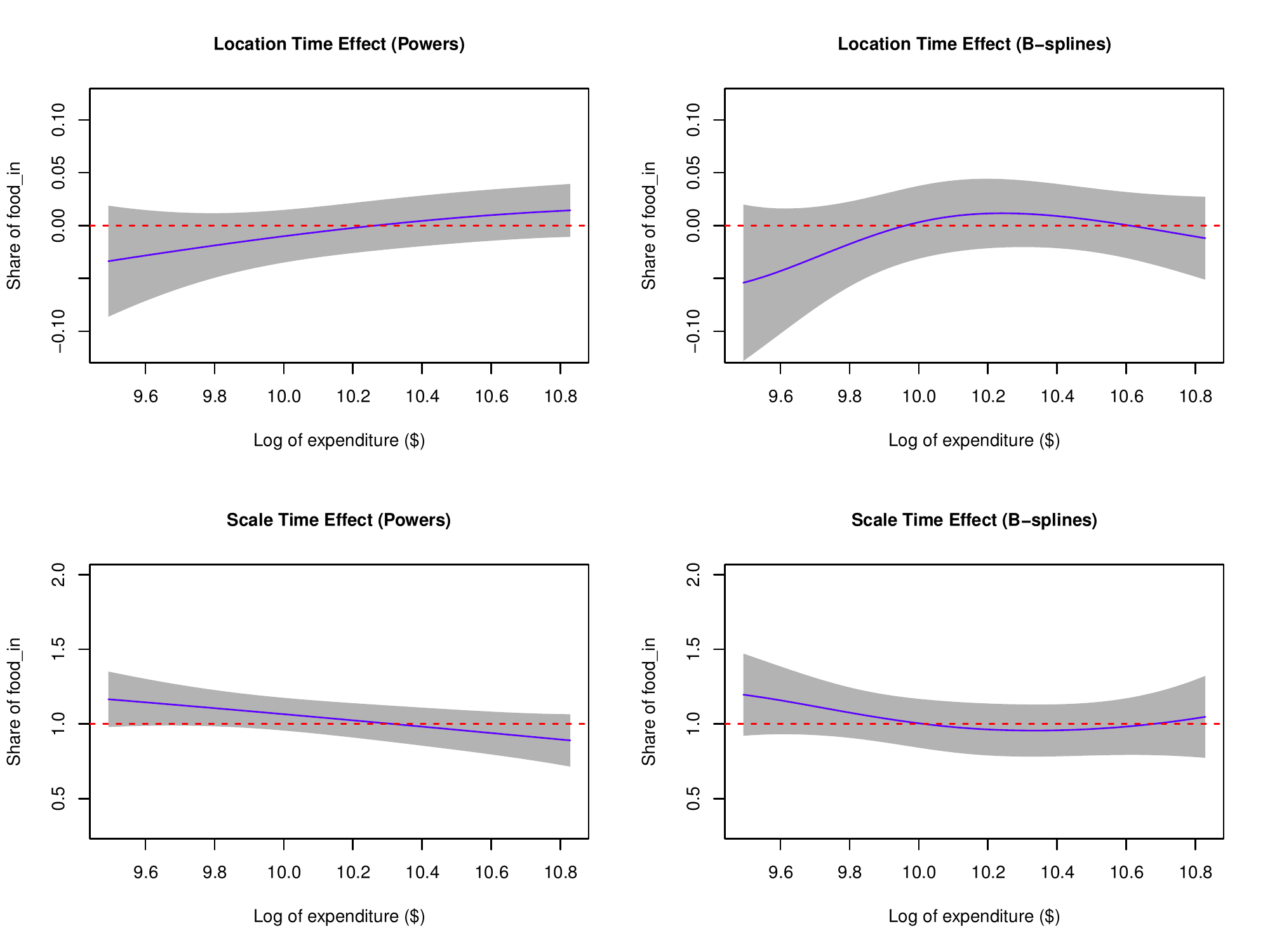,width=6.5in,height=6.5in,angle=-90}

\caption{\label{fig:te2_asf_food_in}Location and scale time effects for food at home share: estimates from conditional quantiles.}

\end{center}

\end{figure}



\begin{figure}

\begin{center}

\centering\epsfig{figure=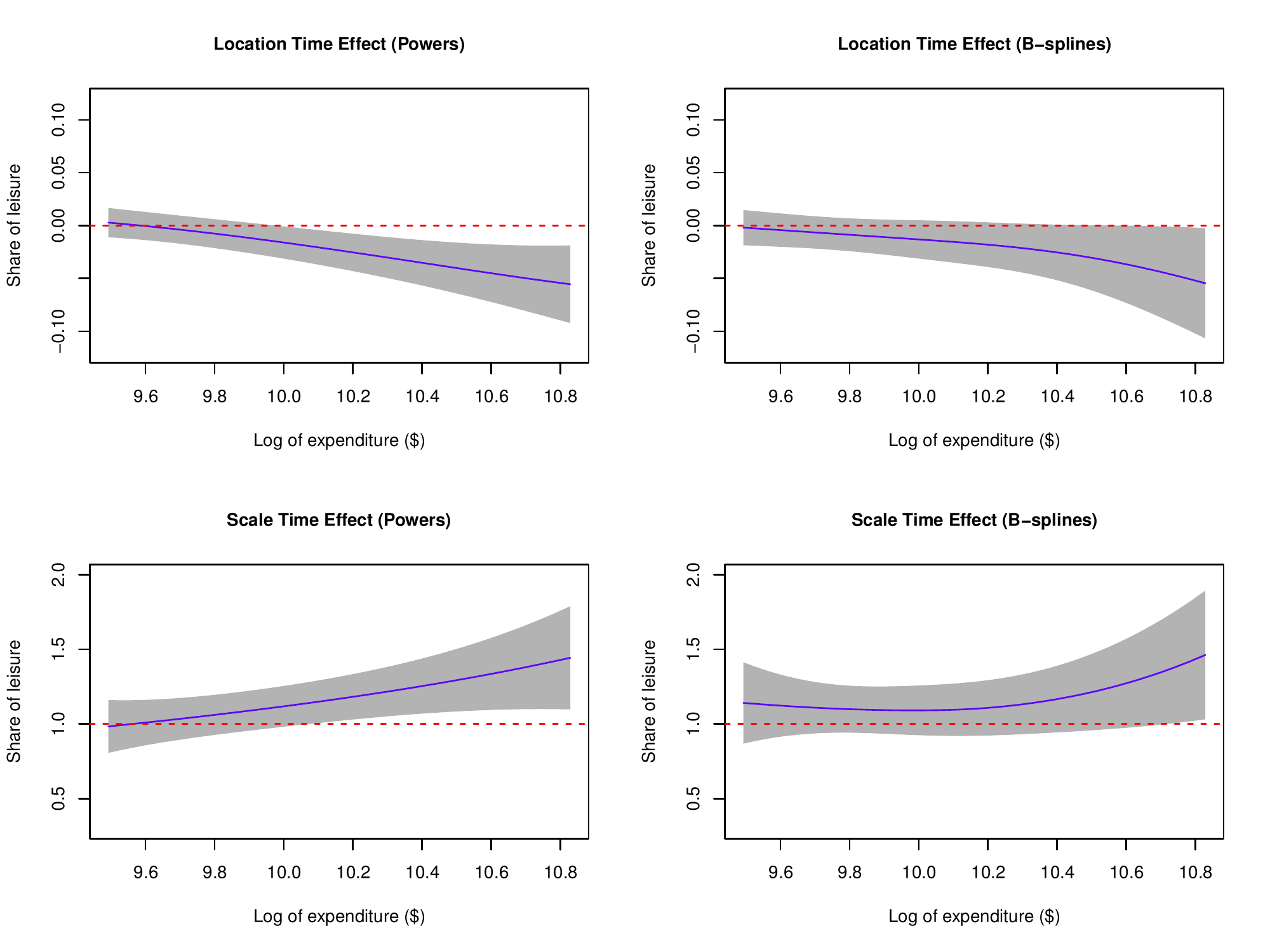,width=6.5in,height=6.5in,angle=-90}

\caption{\label{fig:te2_asf_leisure} Location and scale time effects for leisure share: estimates from conditional quantiles.}

\end{center}

\end{figure}



\begin{figure}

\begin{center}

\centering\epsfig{figure=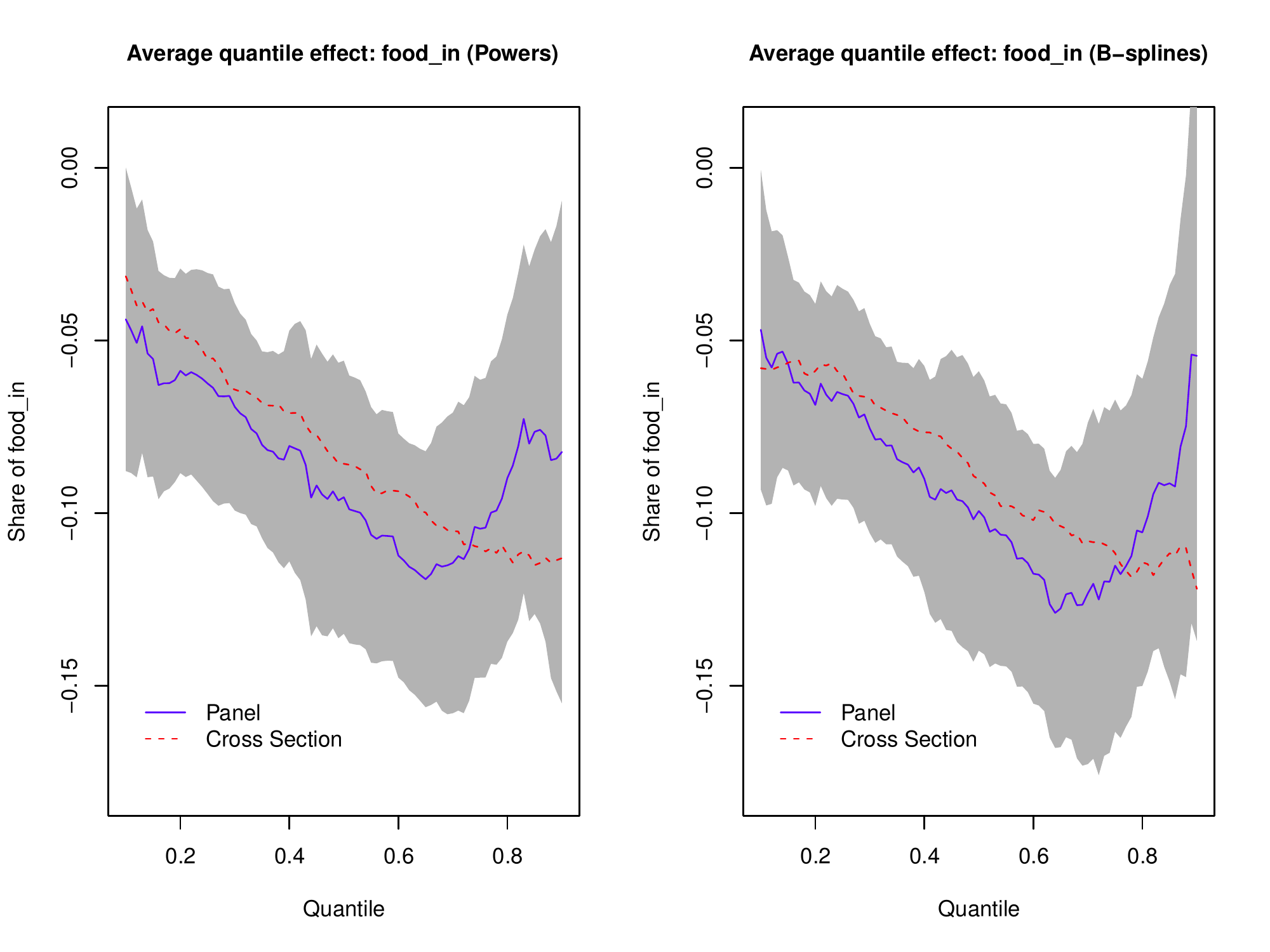,width=3.25in,height=6.5in,angle=-90}

\centering\epsfig{figure=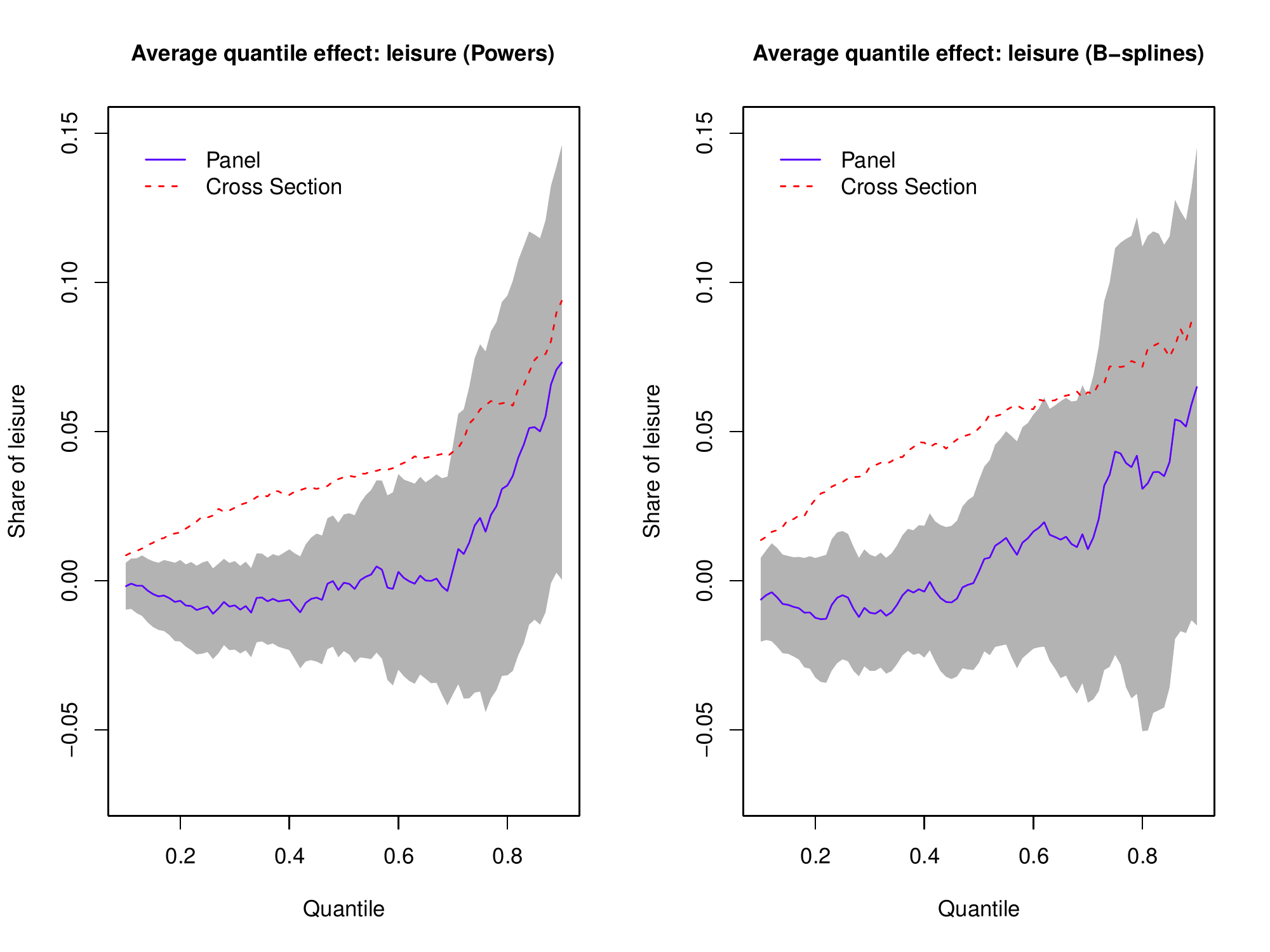,width=3.25in,height=6.5in,angle=-90}

\caption{\label{fig:dqsf}Average conditional quantile effects of log total expenditure.}

\end{center}

\end{figure}



\begin{figure}

\begin{center}

\centering\epsfig{figure=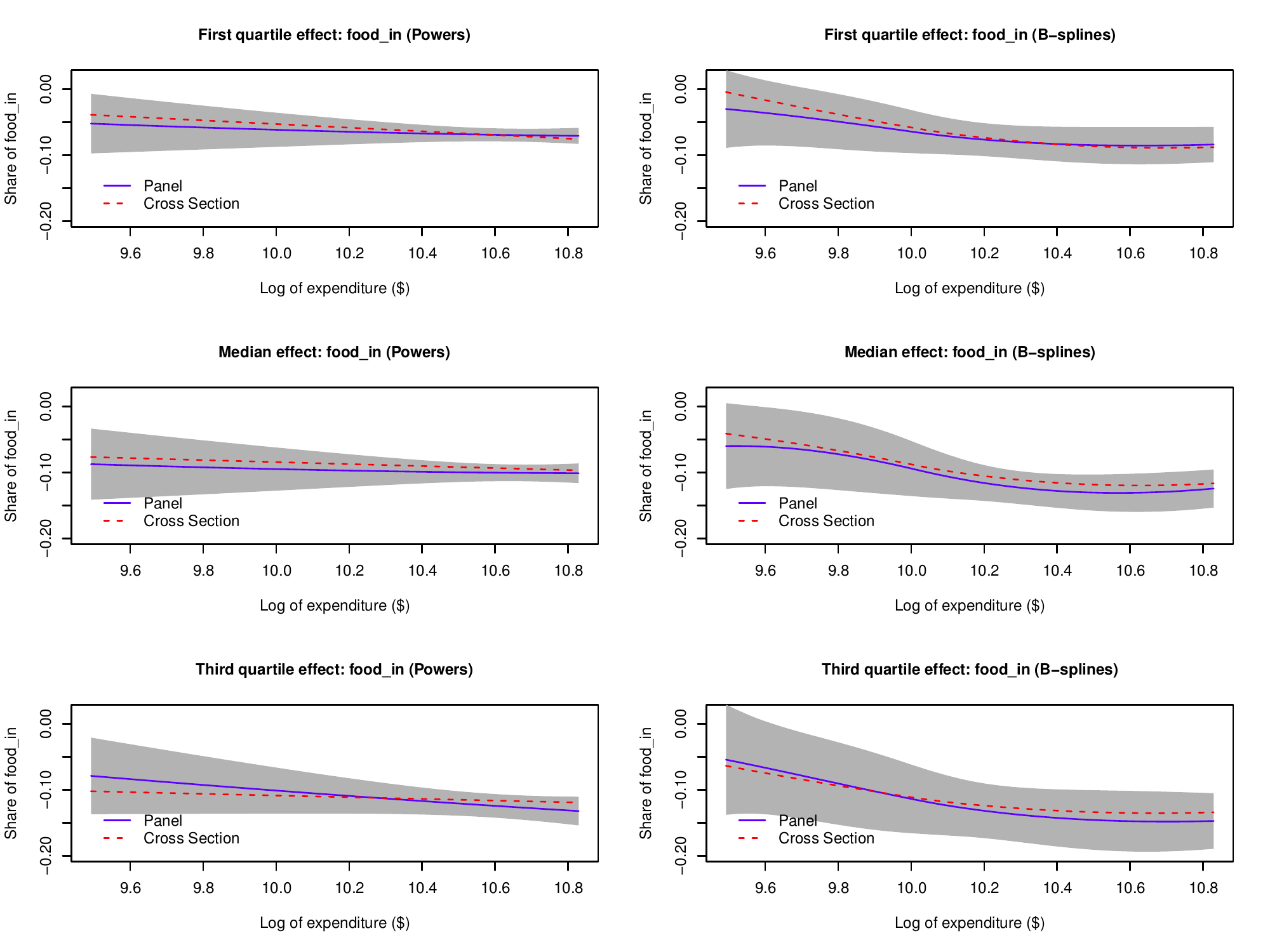,width=8.5in,height=6.5in,angle=-90}

\caption{\label{fig:dqsfs_food_in} Conditional quartile effects of log  total expenditure on food at home share.}

\end{center}

\end{figure}



\begin{figure}

\begin{center}

\centering\epsfig{figure=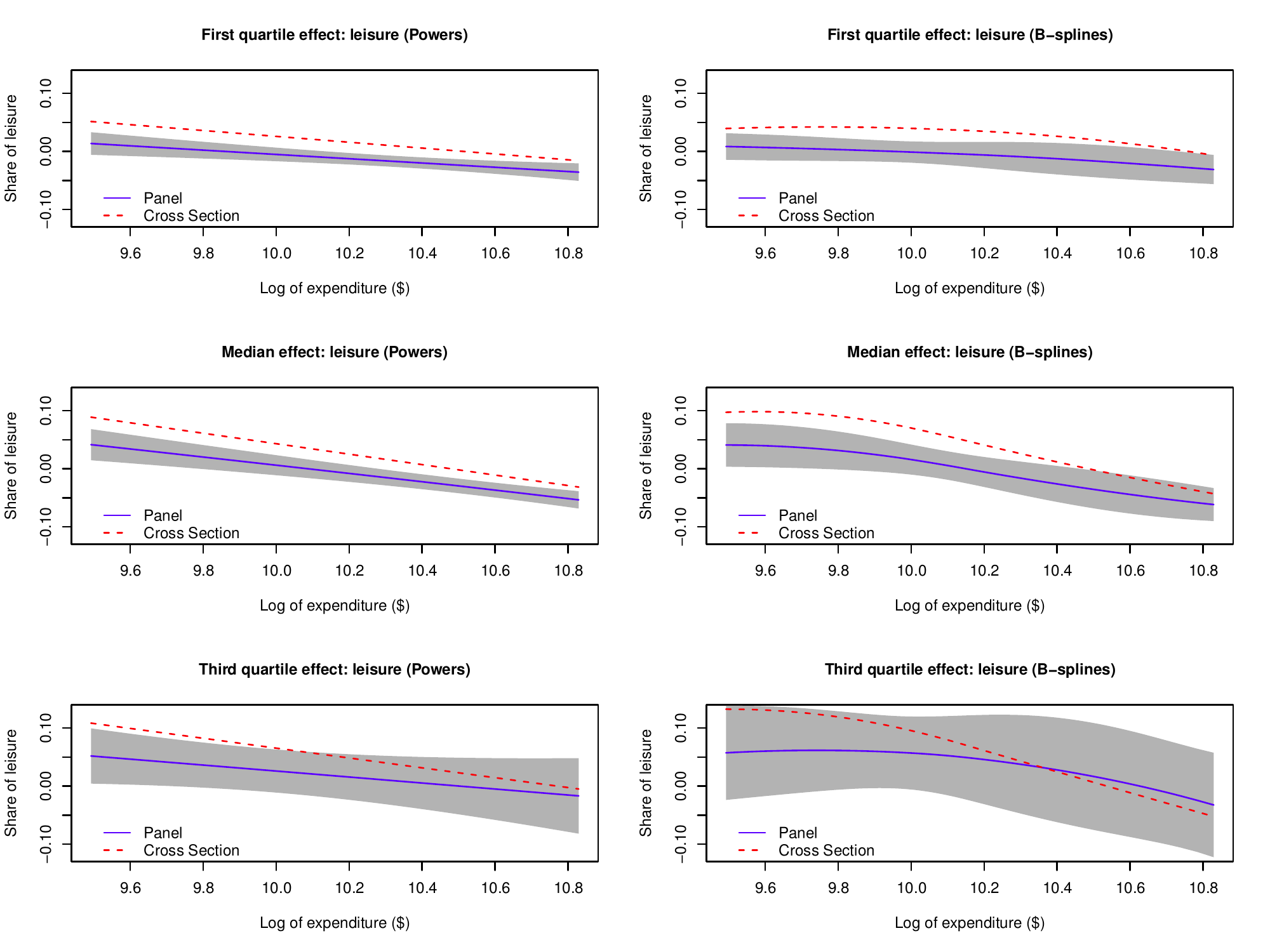,width=8.5in,height=6.5in,angle=-90}

\caption{\label{fig:dqsfs_leisure} Conditional quartile effects of log  total expenditure on leisure share.}

\end{center}

\end{figure}



\begin{figure}

\begin{center}

\centering\epsfig{figure=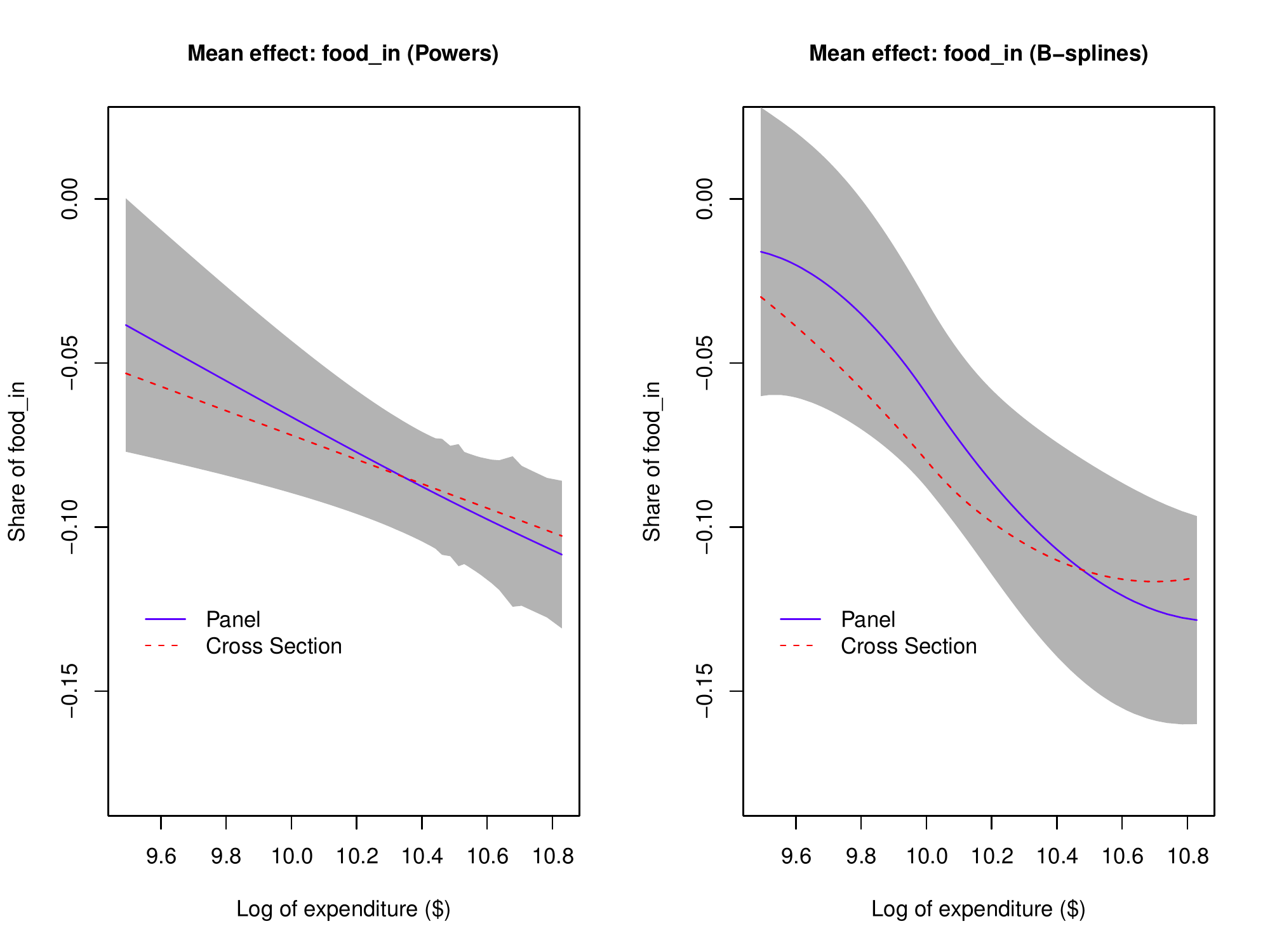,width=3.25in,height=6.5in,angle=-90}

\epsfig{figure=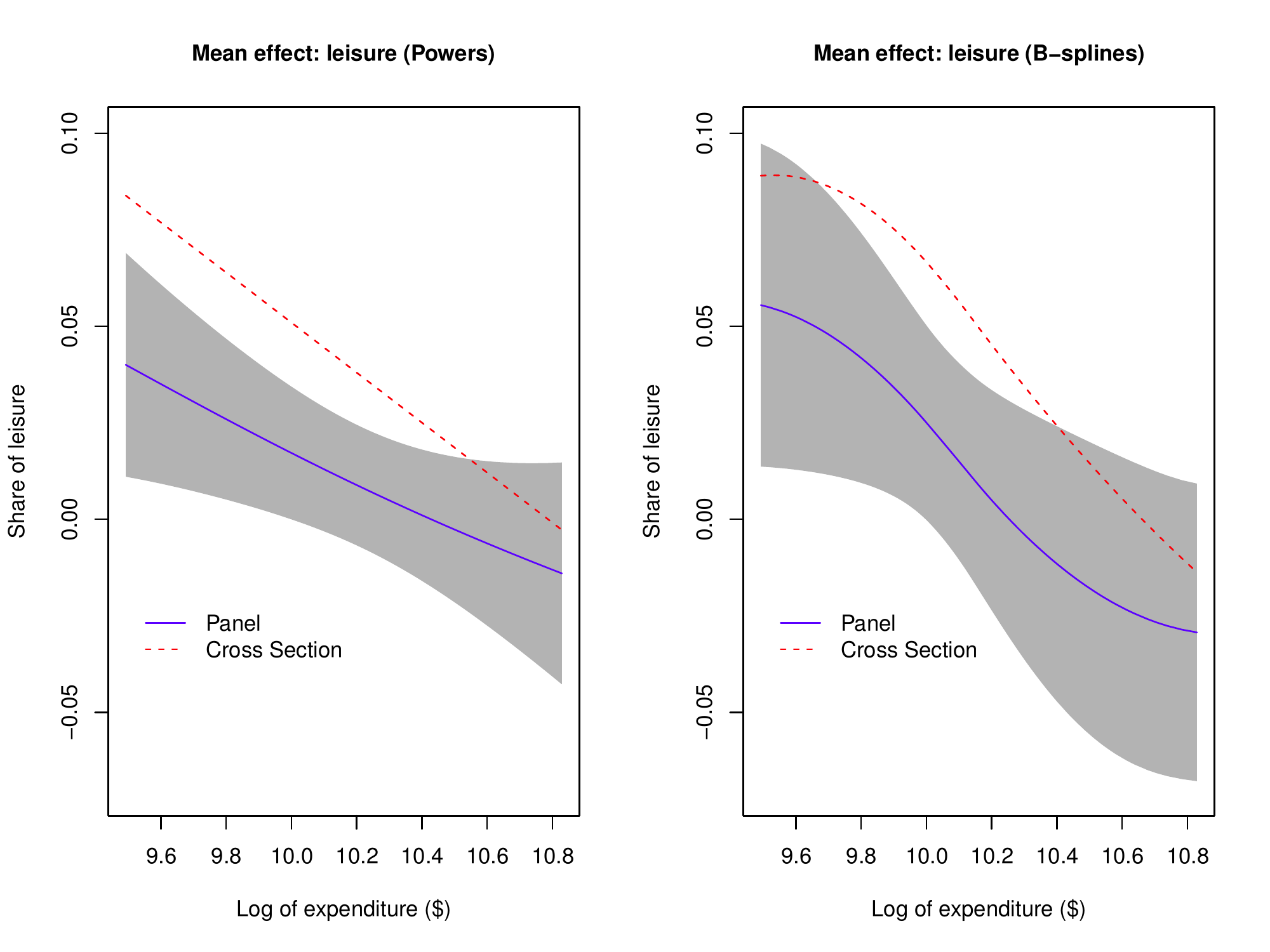,width=3.25in,height=6.5in,angle=-90}

\caption{\label{fig:dasf} Conditional mean effects of log total expenditure.}

\end{center}

\end{figure}


\end{document}